\documentclass[]{llncs}

\usepackage{url}

\newcommand{\sort}[1]{\texttt{#1}}
\newcommand{\Eq}{\mathcal{E}}
\newcommand{\mc}[1]{\mathcal{#1}}

\usepackage{cite}
\usepackage{graphicx}
\usepackage{tikz}
\usetikzlibrary{automata, arrows.meta, positioning}

\usepackage{etoolbox,environ}
\newtoggle{short}
\togglefalse{short}

\usepackage{listings}
\usepackage{algorithm}
\usepackage[noend]{algpseudocode}

\usepackage{stmaryrd}
\usepackage{amssymb}
\usepackage{amsmath}
\usepackage{amsfonts}
\usepackage{mathtools}

\newcommand{\N}{\mathbb{N}}
\newcommand{\TES}[1]{\mathit{TES(#1)}}

\newcommand{\Po}{\mathcal{P}}
\newcommand{\Rp}{\mathbb{R}_+}
\newcommand{\Rel}{\mathcal{R}}
\newcommand{\pr}{\mathrm{pr}}

\newcommand{\id}{\mathit{id}}

\newcommand{\Linf}[1]{\mathcal{L}^\mathrm{inf}(#1)}
\newcommand{\Lfin}[1]{\mathcal{L}^\mathrm{fin}(#1)}
\newcommand{\Lfins}[1]{\mathcal{L}^\mathrm{fin*}(#1)}

\begin{document}
\begin{frontmatter}
\title{A Rewriting Framework for Interacting Cyber-Physical Agents}
\author{Benjamin Lion{1} \and Farhad Arbab\inst{1,2} \and Carolyn Talcott\inst{3}}
\institute{Leiden University, Leiden, The Netherlands \email{lion@cwi.nl}
           \and
            CWI, Amsterdam, The Netherlands \email{arbarb@cwi.nl}
            \and
            SRI International, CA, USA \email{talcott@gmail.com}}
            \maketitle

\begin{abstract}
    The analysis of cyber-physical systems (CPS) is challenging due to the large state space and the continuous changes occurring in their constituent parts. Design practices favor modularity to help reducing this complexity.
    In a previous work, we proposed a discrete semantic model for CPS that captures both cyber and physical aspects as streams of discrete observations, which ultimately form the behavior of a component. 
    This semantic model is denotational and compositional, where each composition operator algebraically models an interaction between a pair of components.

    In this paper, we propose a specification of components as rewrite systems. The specification is operational and executable, and we study conditions for its semantics as components to be compositional. 
    We demonstrate our framework by modeling a coordination of robots moving on a shared field.
    We show that our system of robots can be coordinated by a protocol in order to exhibit a desired emerging behavior.
    We use an implementation of our framework in Maude to give practical results.
\end{abstract}
\end{frontmatter}

\section{Introduction}
Cyber-physical systems are inherently concurrent.
From a cyber point of view, the timing of a decision to sense or act on its physical environment impacts the resulting outcome.
Moreover, several cyber entities may share the same physical environment, leading to race conditions.
From a physical point of view, the ordering of events is not always possible, as some events may be independent.
Moreover, two observers of the same physical phenomenon may order events differently. 
A concurrency protocol encapsulates the orderings of events acceptable to an application, and expressing protocols as separate, concrete modules (as in exogenous coordination~\cite{DBLP:conf/birthday/Arbab11}) helps to reduce the complexity in the design of cyber-physical systems.

More specifically, in this context, each part of a cyber-physical system (e.g., a car, a road, a battery, etc.) is represented as a module, and the system captures the concurrent and interactive execution of each module. We list the following benefits of such approach.
First, it makes concurrency explicit at the level of modules, amenable to exogenous coordination, which provides the opportunity to reason about concurrency protocols directly as first-class objects (e.g., how much a move of a robot consumes energy, can two robots move `simultaneously', etc.). 
Then, the representation of a system remains small. Often, a modular design allows composing constituent components statically to analyze the resulting system, or dynamically at runtime to keep the state space small for, e.g., simulating some runs.
Finally, a component comes with a notion of an interface, that specifies what is visible and what is hidden from other components.
This way, both discrete and continuous aspects of components have the same type of interface, containing the set of observations over time. 

In~\cite{DBLP:journals/corr/abs-2110-02214} we present a model of components that captures timed-event sequences (TESs) as instances of their behavior. An observation is a set of events with a unique time stamp. A component has an interface that defines which events are observable, and a behavior that denotes all possible sequences of its observations (i.e., a set of TESs). 
Our component model is equipped with a family of operators parametrized with an interaction signature. Thus, cyber-physical systems are defined modularly, where each product of two components models the interaction occurring between the two components. 
The strength, as well as practical limitation, of our semantic model is its abstraction: there is no fixed machine specification that generates the behavior of a component. 
We give in this paper an operational description of components as rewrite systems.

    Rewriting logic is a powerful framework to model concurrent systems~\cite{DBLP:journals/jlp/Meseguer12, DBLP:journals/tcs/Meseguer92}. 
    Moreover, implementations, such as Maude~\cite{DBLP:conf/maude/2007}, make system specifications both executable and analyzable. 
    Rewriting logic is suitable for specifying cyber-physical systems, as the underlying equational theory can represent both discrete and continuous changes. 
    We give an operational specification for components as rewriting systems, and show its compositionality under some assumptions.

    Finally, we apply our work to an example that considers two energy sensitive robots moving on a shared field. 
    Each of the two robots aims at reaching the other robot's initial position which, by symmetry, may eventually lead to a crossing situation. The crossing of the two robots is the source of a livelock behavior which can lead to failure (i.e., no energy left in the battery).
    We show how, an exogenous coordination imposed by a protocol can coordinate the moves of the two robots to avoid the livelock situation. We demonstrate the result using our implementation of our framework in Maude. 

We present the following contributions:
\begin{itemize}
    \item an operational specification of components as rewrite systems;
    \item some conditions for the rewrite system's semantics to be compositional;
    \item an incremental, runtime implementation of composition;
    \item illustration of how a composed Maude specification can be used to incrementally analyze a system design using a case study involving the behavior of two coordinated robot agents roaming on a field.
\end{itemize}

    The remainder of the paper is organized as follows.
    \iftoggle{short}{In Section~\ref{sec:ex}, we explain with an abstract example the feature of our operational framework for components.}
    In Section~\ref{sec:component}, we recall some results on the algebra of components defined in~\cite{DBLP:journals/corr/abs-2110-02214}, and give as examples the component version of a robot, a battery, and their product.
    In Section~\ref{sec:specification}, we give an operational specification, using rewriting logic, of a product of components as a system of agents. We show compositionality: the component of a system of agents is equal to the product of each agent component. 
    In Section~\ref{sec:implementation}, we detail the implementation in Maude of the operational specification given in Section~\ref{sec:specification} and analyse a system consisting of two robots, two private batteries, and a shared field. 
 
\newcommand{\Id}{\mathbb{I}}
\newcommand{\rread}{\mathit{read}}
\newcommand{\acts}{\mathit{acts}}
\newcommand{\move}{\mathit{move}}

\iftoggle{short}{
\section{A system of cyber-physical agents: an example}
\label{sec:ex}

This section illustrates our approach on an intuitive and simple cyber-physical system consisting of two robots roaming on a shared field.
A robot exhibits some cyber aspects, as it takes discrete actions based on its readings. Every robot interacts, as well, with a shared physical resource as it moves around. 
The field models the continuous response of each action (e.g., read or move) performed by a robot. 
A question that will motivate the paper is: given a strategy for both robots (i.e., sequence of moves based on their readings), will both robots, sharing the same physical resource, achieve their goals? If not, can the two robots, without changing their policy, be externally coordinated towards their goals?

In this paper, we specify components in a rewriting framework in order to simulate and analyze their behavior.
In this framework, an \emph{agent}, e.g., a robot or a field, specifies a component as a rewriting theory.
A \emph{system} is a set of agents that run concurrently.
The equational theory of an agent defines how the agent states are updated, and may exhibit both continuous and discrete transformations. The dynamics is captured by rewriting rules and an equational theory at the system level that describes how agents interact.
In our example, for instance, each move of a robot is synchronous with an effect on the field. 
Each agent therefore specifies how the action affects its state, and the system specifies which composite actions (i.e., set of simultaneous actions) may occur. 
We give hereafter an intuitive example that abstracts from the underlying algebra of each agent.

\paragraph{Agent} 
A robot and a field are two examples of an agent that specifies a component as a rewriting theory.
The dynamics of both agents is captured by a rewrite rule of the form:
\[
      (s,\emptyset)\Rightarrow (s',\acts)
\]
where $s$ and $s'$ are state terms, and $\acts$ is a set of actions that the field or the robot proposes as alternatives.
Given an action $a\in \acts$ from the set of possibilities, a function $\phi$ updates the state $s$ and returns a new state $\phi(s',a)$.
The equational theory that specifies $\phi$ may capture both discrete and continuous changes.
The robot and the field run concurrently in a system, where their actions may interact.

\begin{example}[Battery]
  A battery is characterized by a set of internal physical laws that describe the evolution of its energy profile over time under external stimulations. We consider three external stimuli for the battery as three events: a charge, a discharge, and a read event. Each of those events may change the profile of the battery, and we assume that in between two events, the battery energy follows some fixed internal laws.\\
  Formally, we model the energy profile of a battery as a function $f : \Rp \to [0,100\%]$ where $f(t) = 50\%$ means that the charge of the battery at time $t$ is of $50\%$. 
  In general, $f$ may be arbitrarily complex, and captures the response of event occurrences (e.g., charge, discharge, read) and passage of time coherently with the underlying laws (e.g., differential equation). For instance, a charge (or discharge) event at a time $t$ coincides with a change of slope in the function $f$ after time $t$ and before the next event occurrence.\\
  For simplicity, we consider a battery for which $f$ is piecewise linear in between any two events.
  The slope changes according to some internal laws at points where the battery is used for charge or discharge.\\
  In our model, a battery interacts with its environment only at discrete time points.
Therefore, we model the observables of a battery as a function $l : \N \to [0,100\%]$ that intuitively samples the state of the battery at some monotonically increasing and non-Zeno sequence of timestamp values.
We capture, in Definition~\ref{def:component}, the continuous profile of a battery as a component whose behavior contains all of such increasing and non-Zeno sampling sequences for all continuous functions $f$.
\end{example}
\begin{example}[Robot]
    A robot's state contains the previously read values of its sensors.
    Based on its state, a robot decides to move in some specific direction or read its sensors.\\
    Similarly to the battery, we assume that a robot acts periodically at some discrete points in time, such as the sequence $\move(E)$ (i.e., moving East) at time $0$, $\rread((x,y),l)$ (i.e., reading the position $(x,y)$ and the battery level $l$) at time $T$, $\move(W)$ (i.e., moving West) at time $3T$ while doing nothing at time $2T$, etc. The action may have as effect to change the robot's state: typically, the action $\rread((x,y),l)$ updates the state of the robot with the coordinate $(x,y)$ and the battery value $l$.
\end{example}

\paragraph{System}
A system is a set of agents together with a composability constraint $\kappa$ that restricts their updates. For instance, take a system that consists of a robot $\id$ and a field $F$. The concurrent execution of the two agents is given by the following system rewrite rule:
\[
    \{(s_\id,\acts_\id), (s_F, \acts_F)\} \Rightarrow_S  \{(\phi_\id(s_\id,a_\id),\emptyset), (\phi_F(s_F,a_F),\emptyset)\}
\]
where $a_\id \in \acts_\id$ and $a_F \in \acts_F$ are two actions related by $\kappa$.

Each agent is unaware of the other agent's decisions. The system rewrite $\Rightarrow_S$ filters actions that do not comply with the composability relation $\kappa$. As a result, each agent updates its state with the (possibly composite) action chosen at runtime, from the list of its submitted actions.
The framework therefore clearly separates the place where agent's and system's choices are handled, which is a source of runtime analysis.

Already, at this stage, we can ask the following query on the system: will robot $\id$ eventually reach the location $(x,y)$ on the field? Note that the agent alone cannot answer the query, as the answer depends on the characteristics of the field.
\begin{example}[Battery-Robot]
    \label{ex:battery-robot}
    Typically, a move of the robot \emph{synchronizes} with a change of state in the battery, and a read of the robot occurs at the same time as a sampling of the battery value.\\
    The system behavior therefore consists of sequences of simultaneous events occurring between the battery and the robot. 
    By composition, the battery exposes the subset of its behavior that conforms to the specific frequency of read and move actions of the robot. The openness of the battery therefore is reflected by its capacity to adapt to any observation frequency.
\end{example}

\paragraph{Coordination}
Consider now a system with three agents: two robots and a field. Each robot has its own objective (i.e., location to reach) and strategy (i.e., sequence of moves). Since both robots share the same physical field, some exclusion principals apply, e.g., no two robots can be at the same location on the field at the same time.
It is therefore possible that the system deadlocks if no actions are composable, or livelocks if the robots enter an infinite sequence of repeated moves.

We add a protocol agent to the system, which imposes some coordination constraints on the actions performed by robots $\id_1$ and $\id_2$.
Typically, a protocol coordinates robots by forcing them to do some specific actions.
As a result, given a system configuration $\{(s_{\id_1},\acts_{\id_1}), (s_{\id_2},\acts_{\id_2}), (s_F,\acts_F), (s_P, \acts_P)\}$ the run of robots $\id_1$ and $\id_2$ has to agree with the observations of the protocol, and the sequence of actions for each robot will therefore be conform to a permissible sequence under the protocol.

In the case where the two robots enter a livelock and eventually run out of energy, we show in Section~\ref{sec:implementation} the possibility of using a protocol to remove such behavior. 
\begin{example}[Safety property]
  A safety property is typically a set of traces for which nothing bad happens.
  In our framework, we consider only observable behaviors, and a safety property therefore declares that nothing bad \emph{is observable}. However, it is not sufficient for a system to satisfy a safety property to conclude that it is safe: an observation that would make a sequence violate the safety property may be absent, not because it did not actually happen, but merely because the system missed to detect it.\\
  For example, consider a product of a battery component and a robot with a sampling period $T$, as introduced in Example~\ref{ex:battery-robot}. 
  Consider the safety property: \emph{the battery energy is between the energy thresholds $e_1$ and $e_2$}. 
  The resulting system may exhibit observations with energy readings between the two thresholds only, and therefore satisfy the property. However, had the robot used a smaller sampling period $T'=T/2$, which adds a reading observation of its battery between every two observations, we may have been able to detect that the system is not safe because it produces sequences at this finer granularity sampling rate that  violate the safety property.
  We show how to algebraically capture the safety of a system constituted of a battery-robot. 
\end{example}
}
\section{Semantic model: algebra of components}
\label{sec:component}
The design of complex systems becomes simpler if such systems can be decomposed into smaller sub-systems
that interact with each other. In order to simplify the design of cyber-physical systems, we introduced in~\cite{DBLP:journals/corr/abs-2110-02214} a semantic model that abstracts from the internal
details of both cyber and physical processes. As first class entities in this model, a component
encapsulates a behavior (set of TESs) and an interface (set of events).
We recall basic definitions and properties in this section. \iftoggle{short}{}{See \ref{appendix:example} for additional examples.}

\subsection{Components }
\label{section:components}

\paragraph{Preliminaries}
A timed-event stream, TES, $\sigma$ over a set of events $E$ is an infinite sequence of \emph{observations}, where its $i^{th}$ observation $\sigma(i) = (O, t)$, $i \in \N$, consists of a pair of a subset of events in $O \subseteq E$, called the \emph{observable}, and a positive real number $t \in \Rp$ as time stamp.
A timed-event stream (TES) has the additional properties that its consecutive time stamps are monotonically increasing and non-Zeno, i.e., if $\sigma(i) = (O_i, t_i)$ is the $i^{th}$ element of TES $\sigma$, then (1) $t_i < t_{i+1}$, and (2) for any time $t \in \Rp$, there exists an element $\sigma(i) = (O_i, t_i)$ in $\sigma$ such that $t< t_i$.
We use $\sigma^{(k)}$ to denote the $k$-th derivative of the stream $\sigma$, such that $\sigma^{(k)}(i) = \sigma(i+k)$ for all $i \in \N$.
We refer to the stream of observables of $\sigma$ as its first projection $\pr_1(\sigma) \in \Po(E)^\omega$, and the stream of time stamps as its second projection $\pr_2(\sigma) \in \Rp^\omega$.
We write $(O,t) \in \sigma$ if there exists $i \in \N$ such that $\sigma(i) = (O,t)$.

\newcommand{\dom}{\mathit{dom}}
We write $\sigma(t) = O$ if there exists $i \in \N$ such that $\sigma(i) = (O,t)$, and $\sigma(t) = \emptyset$ otherwise. We use $\it{dom}(\sigma)$ to refer to the set of observable time stamps, i.e., the set $\it{dom}(\sigma) = \{ t\in\Rp \mid \exists i. \pr_2(\sigma)(i)=t\}$. Moreover, we use $\sigma\cup\tau$ to denote the stream such that, for all $t\in\Rp$, $(\sigma \cup \tau)(t) = \sigma(t) \cup \tau(t)$ and $\dom(\sigma \cup \tau) = \dom(\sigma) \cup \dom(\tau)$

A component denotes \emph{what} observables are possible, over time, given a fixed set of events. We give three examples of components, which capture some cyber-physical aspects of concurrent systems. 

\begin{definition}[Component]
     \label{def:component}
     A \emph{component} $C = (E,L)$ is a pair of a set of events $E$, called its \emph{interface}, and a \emph{behavior} $L \subseteq \TES{E}$.
\end{definition}
Given component $A = (E_A, L_A)$, we write $\sigma:A$ for a TES $\sigma \in L_A$.
    \begin{example}[Battery]
        A battery component is a pair $(E_B(C),L_B(C))$ with events $\rread(l) \in E_B$ for $0\% \leq l\leq 100\%$, $\textit{charge}(\mu) \in E_B$, and $\textit{discharge}(\mu) \in E_B$ with $\mu$ a (dis)charging coefficient in $\%$ per seconds,
         and $C$ a constant capacity in $\textit{mAH}$. 
The battery displays its capacity with the event $\textit{capacity(C)}$.
        The behavior $L_B$ is a set of sequences $\sigma \in L_B$ such that there exists a piecewise linear function $f:\Rp \to \Po(E_B)$ with, for $\sigma(i) = (O_i,t_i)$,  
        \begin{itemize}
            \item for $\sigma(0) = (O_0,t_0)$, $f([0;t_0]) = 100\%$, i.e., the battery is initially fully charged;
            \item if $O_i = \{\rread(l)\}$, then $f(t_i) = l$ and the derivation $f'_{[t_{i-1},t_{i+1}]}$ of $f$ is constant in $[t_{i-1},t_{i+1}]$, i.e., the observation does not change the slope of $f$ at time $t_i$;
            \item if $O_i =\{\textit{discharge}(\mu)\}$, then $f_{[t_i,t_{i+1}]}(t) = \max(f(t_i) - (t - t_i)\mu, 0)$;
            \item if $O_i =\{\textit{charge}(\mu)\}$, then $f_{[t_i,t_{i+1}]}(t) = \min(f(t_i) + (t - t_i)\mu, 100)$;
        \end{itemize}
        where $f_{[t_1;t_2]}$ is the restriction of function $f$ on the interval $[t_1;t_2]$.
        There is \emph{a priori} no restrictions on the time interval between two observations, as long as the sequence of timestamps is increasing and non-Zeno.
    $\square$
    \end{example}
\begin{example}[Robot]
    A robot with identifier $i$ is a component $R(i,T) = (E_R,L_R(T))$ with events $\rread(i,l) \in E_R$ for $0\% \leq l\leq 100\%$, $\textit{d}(i,p) \in E_R$ with $p$ the power requested by the robot for the move and $d$ the direction, and $T$ a period in seconds.
    For instance, the event $N(i,p)$ represents robot $i$ moving $N$orth with power $p$. 
    The robot reads the capacity of its battery with the event $\textit{getCapacity(i,C)} \in E_R$, with $C$ in $\textit{mAH}$. Once the robot knows the capacity of the battery, the values read in percent can be converted to remaining power.

    The behavior $L_R(T)$ contains any sequence of observations at fix period $T$, such that $\sigma \in L_R(T)$ if and only if $\sigma(i) = (O_i, t_i)$ implies $t_i = kT$ with $k\in\N$ and $O_i \subseteq E_R$ with $|O_i|=1$.
    We assume that the robot does one action at a time: either a read of its sensors, or a move in some direction.
    $\square$
\end{example}

\iftoggle{short}{
\begin{example}
    Most of cyber systems, such as embedded systems, run in sync with their clock.
    The clock therefore dictates how fast a sensor can be read, or how many decisions can be taken in a time interval. 
    The component version of such an embedded system has, as timing constraint, that each observable is time-stamped according to the clock frequency.
    As a result, each TES in its behavior has time-stamps that are multiples of the period $T$.

    \hfill$\blacksquare$
\end{example}
\begin{example}
    Physical parts of systems, such as batteries, heaters, fans, etc. provide a service that must satisfy some qualities.
    A battery delivers an amount of power that drives, for instance, the wheels of a robot; a heater radiates some energy to warm a room; a fan spins in order to cool down a hot piece. In different contexts, physical parts may behave differently. Capturing some use cases formally helps in assessing the quality and the adequacy of a particular system.

    For instance, a battery as a component captures all \emph{significant} (if not all \emph{possible}) sequences of charge/discharge and value readings. While only one of such sequence will be observable for each usage, the analysis of all its behavior is crucial to assess its quality.
    \hfill$\blacksquare$
\end{example}
}{}

\newcommand{\Stop}{\mathit{stop}}
\newcommand{\swap}{\mathit{swap}}

\subsection{Product and division}
\label{sec:algebra}
Components describe which observations occur over time.
When run concurrently, observable events from a component may relate to observable events of another component.
This relation defines what kind of interaction occurs between the two components, as it may enforce two events to occur within the same observable at the same time (e.g., actuation of a wheel and changes of location of the robot), or it may prevent two events to occur simultaneously (e.g., two robots moving to the same physical location). 
Interaction constraints are therefore captured by an algebraic operator that acts on components.
The result of forming the product of two components is a new component, whose behavior contains the composition of every pair of TESs, one from each product operand, that satisfies the underlying constraints imposed by that specific operator.

Let $A = (E_A, L_A)$ and $B=(E_B, L_B)$ be two components. We use the relation $R(E_A, E_B) \subseteq \TES{E_A} \times \TES{E_B}$ and the function $\oplus : \TES{E} \times \TES{E} \rightarrow \TES{E}$, with $E = E_A\cup E_B$, to range over composability relations and composition functions, respectively. We use $\Sigma$ to range over interaction signatures, i.e., pairs of a composability relation and a composition function.
 \begin{definition}[Product]
     \label{def:prod}
     The \emph{product} of components $A$ and $B$ under interaction signature $\Sigma = (R,\oplus)$ is the component $C = A \times_\Sigma B = (E_A \cup E_B, L)$  where
    $$ 
         L = \{ \sigma \oplus \tau \mid \sigma \in L_A, \tau \in L_B,\ (\sigma, \tau) \in R(E_A, E_B)\}    
     $$
 \end{definition}

For simplicity, we write $\times$ as a general product when the specific $\Sigma$ is irrelevant.
\begin{example}
    We define $\Sigma_{RB} = ([\kappa_{RB}],\cup)$ 
    where $\cup$ unions two TESs as defined in the preliminaries, and $[\kappa_{RB}]$ specifies co-inductively (see~\cite{DBLP:journals/corr/abs-2110-02214} for details of the construction), from a relation on observations $\kappa_{RB}$, how event occurrences relate in the robot and the battery components of capacity $C$.
    More specifically, $\kappa_{RB}$ is the smallest symmetric relation over observations such that $((O_1,t_1), (O_2,t_2)) \in \kappa_{RB}$ implies that $t_1=t_2$ and 
    \begin{itemize}
        \item the discharge event in the battery coincides with a move of the robot, i.e., $\textit{d}(i,p) \in O_1$ if and only if $\textit{discharge}(\mu) \in O_2$. Moreover, the interaction signature imposes a relation between the discharge coefficient $\mu$ and the required power $p$, i.e., $\mu = p/C$;
        \item the read value of the robot sensor coincides with a value from the battery component, i.e., $\rread(i,l) \in O_1$ if and only if $\rread(l) \in O_2$;
        \item the robot reads the capacity value that corresponds to the battery capacity, i.e., $\textit{getCapacity(i,c)} \in O_1$ if and only if $\textit{capacity(c)} \in O_2$.
    \end{itemize}
    The product $B \times_{\Sigma_{RB}} R(T,i)$ of a robot and a battery component, under the interaction signature $\Sigma_{RB}$, restricts the behavior of the battery to match the periodic behavior of the robot, and restricts the behavior of the robot to match the sensor values delivered by the battery.\\
    As a result, the behavior of the product component $B \times_{\Sigma_{RB}} R(T,i)$ contains all observations that the robot performs in interaction with its battery. Note that trace properties, such as \emph{all energy sensor values observed by the robot are within a safety interval}, does not necessarily entail safety of the system: some unobserved energy values may fall outside of the safety interval. Moreover, the frequency by which the robot samples may \emph{reveal} some new observations, and such robot can safely sample at period $T$ if, for any period $T'\leq T$, the product $B\times_{\Sigma_{RB}} R(T',i)$ satisfies the safety property. 
    $\square$
\end{example}
 
\iftoggle{short}{
\begin{example}
    Every move operation performed by a robot discharges its battery, and every read operation of a robot gets the value of the battery energy level.
    By composing a robot and a battery component, under an appropriate interaction signature, the new component exposes, for every read operation of the robot, the appropriate battery value.
    Typically, this is formally achieved by defining both robot and battery as components $R$ and $B$, for which the suitable interaction signature $\Sigma$ is given. The resulting product $R\times_\Sigma B$ captures the desired TESs in its behavior.
\end{example}
}{}

\newcommand{\ssync}{\mathit{sync}}
\iftoggle{short}{
Consider two components $B$ and $C$, and a product $\times$ capturing some interaction constraints between $B$ and $C$.
Then, the composite expression $A = C \times B$ captures, as a component, the concurrent observations of components $C$ and $B$ under the interaction modelled by $\times$. 
Consider a component $D$ such that $C \times B = D \times B$. If $D$ is different from $C$, then the equality states that the result of $D$ interacting with $B$ is the same as $C$ interacting with $B$. 
Consequently, in this context, component $C$ can be replaced by component $D$ while preserving the global behavior of $A$.

In general, a component $D$ that can substitute for $C$ is not unique. The set of alternatives for $D$ depends, moreover, on the product $\times$, on the component $B$, and on the behavior of $A$.  A `goodness' measure may induce an order on this set of components, and eventually give rise to a \emph{best} substitution.
More generally, the problem is to characterize, given two components $A$ and $B$ and an interaction product $\times$, the set of all $C$ such that $A = C \times B$.

Naturally, the definition of product comes with a dual decomposition operation. A quotient is a component that, whose product (under the same interaction signature) with the divisor, yields the dividend. 
As there may be several possible decompositions, we consider the set of all such possible quotients.
For simplicity, we assume $\times$ to be commutative. Right and left quotients can be similarly defined when $\times$ is not commutative.
\begin{definition}[Quotients]
    The quotients of $A$ by $B$ under the interaction signature $\Sigma$, written $A/^*_\Sigma B$, is the set $\{C \mid B \times_\Sigma C = A \}$.
\end{definition}
We say that $A$ is \emph{divisible} by $B$ (or $B$ divides $A$) under $\Sigma$ if the set of quotients is not empty.
We define division operators that pick, given a choice function, the best element from their respective sets of quotients as their quotients.
\begin{example}
    Consider a robot that performs $5$ moves, and then stops.
    Each move consumes some energy, and the robot therefore requires sufficient amount of energy to achieve its moves.
    The product of a robot $C$ with its battery $B$ under the interaction signature $\Sigma$ is given by the expression $A =C \times_\Sigma B$, where $\Sigma$ synchronizes a move of robot $C$ with battery $B$. 
    Note that different batteries behave differently.
    The set of batteries that would lead to the same behavior is given by the quotients of $A$ by $B$.
\end{example}

\begin{definition}[Division]
    The division of $A$ by $B$, under the interaction signature $\Sigma$ and the choice function $\chi$ over the quotients, is the element $\chi(A/_\Sigma^* B)$. We write $A/_\Sigma^\chi B$ to represent the division.
\end{definition}
\begin{example}
    It is usual (e.g., \cite{10.1145/3517192}) to consider the greatest common divisor when forming the product of cyber-physical components, so that no observation is missed. 
    Our operation of division, however, gives an alternative perspective.
    Let $C(H)$ be a component whose observations have multiples of $H\in\Rp$ as time-stamps.
    Then, let $A = C(H_1)\times_\Sigma C(H_2)$ with $H_1,H_2\in\Rp$.
    The set of components $\{C(H) \mid A = C(H_1) \times_\Sigma C(H), H \in \Rp\}$ contains all the quotients of $A$ divisible by $C(H_1)$. The selection of the component with the shortest period $H$ would be one choice function for the division of $A$ by $C(H)$ under $\Sigma$.
\end{example}

\begin{lemma}
    Let $\times_\Sigma$ be commutative.
    Given $A$ divisible by $B$ under $\Sigma$ and $\chi$ a choice function on the set of quotients of $A$ divisible by $B$, then $B \in A/_\Sigma^* (A /_\Sigma^\chi B)$.
\end{lemma}
\begin{proof}
    If $A$ is divisible by $B$ under $\Sigma$ and if $\chi$ selects one quotient over the set, then $ C = A/_\Sigma^\chi B$ is such that $A = B \times_\Sigma C$.
    By commutativity of $\times_\Sigma$, $A = C \times_\Sigma B$ and $B \in A/_\Sigma^* C$.
\end{proof}
\begin{lemma}
    Let $\times_\Sigma$ be associative.
    If $A$ is divisible by $B$ under $\Sigma$ and $B$ is divisible by $C$ under $\Sigma$, then $A$ is divisible by $C$ under $\Sigma$.
\end{lemma}
\begin{proof}
    If $A$ is divisible by $B$ under $\Sigma$, then there exists $D$ such that $A = B \times_\Sigma D$.
    If $B$ is divisible by $C$ under $\Sigma$, then there exists $E$ such that $B = C \times_\Sigma E$.
    By substitution, we have $A = (C \times_\Sigma E) \times_\Sigma D$.
    Using associativity of $\times_\Sigma$, we get $A = C \times_\Sigma (E \times_\Sigma D)$ which proves that $A$ is divisible by $C$ under $\times_\Sigma$.
\end{proof}

\begin{example}
    Consider the system 
    \[S(T_1, ..., T_n) =\ \bowtie_{i\in\{1,...,n\}}(R_i(T_i)\times_{\Sigma_{R_iB_i}} B_i) \bowtie_{\times_{RF}} F\] 
    made of $n$ robots $R_i(T_i)$, each interacting with a private battery $B_i$ under the interaction signatures $\Sigma_{R_iB_i}$, and in product with a field $F$ under the interaction signature $\Sigma_{RF}$. We use $\bowtie$ for the product with the free interaction signature (i.e., every pair of TESs is composable), and the notation $\bowtie_{i\in\{1,...,n\}} \{C_i\}$ for $C_1 \bowtie ...\bowtie C_n$ as $\bowtie$ is commutative and associative.\\
Consider the safety property $P_{\mathit{safe}}$, that captures as a component the set of all TESs that has observables of the batteries $B_i$ within two threshold values.\\
\end{example}
}{}

\section{System of agents and compositional semantics}
\label{sec:specification}
    Components in Section~\ref{sec:component} are declarative.
    Their behavior consists of all the TESs that satisfy some internal constraints. 
    The abstraction of internal states in components makes the specification of observables and their interaction easier.
    The downside of such declarative specification lies in the difficulty of generating an element from the behavior, and ultimately verifying properties on a product expression.

     An operational specification of a component provides a mechanism to construct elements in its behavior.
     An \emph{agent} is the operational specification that produces finite sequences of observations that, in the limit, determine the behavior of a component.
     An agent is stateful, and has transitions between states, each labeled by an observation, i.e., a set of events with a time-stamp.
     We consider a finite specification of an agent as a rewrite theory, where finite applications of the agent's rewrite rules generate a sequence of observables that form a prefix of some elements in the behavior of its corresponding component.
     We restrict the current work to integer time labeled observations. 
     While in the cyber-physical world, time is a real quantity, we consider in our fragment a countable infinite domain for time, i.e., natural numbers. 
     The time interval between two tics is therefore the same for all agents, and may be interpreted as, e.g., seconds, milliseconds, femtoseconds, etc.
     We show how an agent may synchronize with a local clock that forbids actions at some time values, thus modeling different execution speeds.

     An operational specification of a composite component provides a mechanism to construct elements in the behavior of a product expression.
     The product on components is parametrized by an interaction signature that tells \emph{which} TESs can compose, and \emph{how} they compose to a new TES.
     We consider, in the operational fragment of this section, interaction signatures each of whose composability relation is co-inductively defined from a relation on observations $\kappa$. Intuitively, such restriction enables a step-by-step operation to check that the head of each sequence is valid, i.e., extends the sequence to be a prefix of some elements in the composite component.
     Moreover, we require $\kappa$ to be such that the product on component $\times_{([\kappa],\cup)}$ is commutative and associative (see~\cite{DBLP:journals/corr/abs-2110-02214}). 
     By \emph{system} we mean a set of agents that compose under some interaction signature $\Sigma = ([\kappa],\cup)$.
     A system is stateful, where each state is formed from the states of its component agents, and has transitions between states, each labeled by an observation, formed from the component agent observations.
     We consider a finite specification of a system as the composition of a set of rewriting theories (one for each agent), and a system rewrite rule that produces a composite observation complying with the relation $\kappa$.
     \iftoggle{short}{}{
         We prove compositionality: the system component is equal to the product under the interaction signature $\Sigma = ([\kappa],\cup)$ of every one of its constituent agent components.}

     \iftoggle{short}{
     Agents and systems are both operational specifications of components. 
     An agent component consists of all TESs for which every finite prefix is generated by a finite sequence of its respective agent's rule applications.
     The system component consists of all TESs for which every finite prefix is generated by a finite sequence of the agents' rule applications while conforming to the composability relation $\kappa$.
     We prove compositionality: the system component is equal to the product under the interaction signature $\Sigma = ([\kappa],\cup)$ of every one of its constituent agent components.
 }{}

    In order to give to the agent a semantics as components, we recall some results and notations about TES transition systems $T = (Q,E,\to)$ (see \cite{https://doi.org/10.48550/arxiv.2205.13008} and \ref{appendix:TTS} for more results on TES transition systems)  where $Q$ is a set of states, $E$ a set of events, and $\to \subseteq Q \times (\Po(E)\times \Rp)\times Q$ a set of transitions.
    
    We write $q\xrightarrow{u} p$ for the sequence of transitions $q\xrightarrow{u(0)}q_1\xrightarrow{u(1)}q_2...\xrightarrow{u(n-1)}p$, where $u = \langle u(0), ..., u(n-1)\rangle \in(\Po(E)\times \Rp)^n$. We write $|u|$ for the size of the sequence $u$.

We use $\Lfin{T,q}$ to denote the set of finite sequences of observables labeling a finite path in $T$ starting from state $q$, such that 
\[\Lfin{T,q} =\{u \mid \exists q'. q\xrightarrow{u} q', \forall i< |u|-1. u(i) = (O_i,t_i) \land t_i < t_{i+1} \}\]
Additionally, the set $\Lfins{T,q}$ is the set of sequences from $\Lfin{T,q}$ postfixed with empty observations, i.e., the set 
\[\Lfins{T,q} = \{u\tau \in \TES{E}{} \mid u \in \Lfin{T,q} \textit{ and } \tau \in \TES{\emptyset}{}\}\]
We use $\Linf{T,q}$ to denote the set of TESs labeling infinite paths in $T$ starting from state $q$, such that \[\Linf{T,q} =\{\sigma\in \TES{E}\mid \forall n.\sigma[n]\in \Lfin{T,q}\}\] where, as introduced in Section~\ref{sec:component}, $\sigma[n]$ is the prefix of size $n$ of $\sigma$.
\newcommand{\cl}[1]{\mathit{cl}(#1)}
\newcommand{\tr}[1]{\mathcal{T}_{#1}}
\newcommand{\sem}[1]{\llbracket #1 \rrbracket}

Let $X\subseteq \TES{E}{}$, we use $\cl{X}$ 
to denote the set that contains the continuation with empty observations of any prefix of an element in $X$, i.e., 
$\cl{X} = \{ u\tau \in \TES{E}{} \mid \tau \in \TES{\emptyset}{}\ \ \it{and}\ \ \exists \sigma.\exists i. \sigma \in X \land \sigma[i] = u\}$.
Given a component $C = (E,L)$, we write $\cl{C}$ for the new component $(E,\cl{L})$.

\subsection{Action, agent, and system}
    We give the operational counterparts of an observation, a component, and a product of components as, respectively, an action, an agent, and a system of agents. See \ref{appendix:proof} for proof sketches.

    \paragraph{Action}
    Actions are terms of sort \sort{Action}. An action has a name of sort \sort{AName} and some parameters. 
    We distinguish two typical actions, the idle action $\star$ and the ending action \sort{end}.
    A term of sort \sort{Action} corresponds to an observable, i.e., a set of events.
    The idle action $\star$ and the ending action \sort{end} both map to the empty set of events.
    An example of an action is \sort{move(R1,d)} or \sort{read(R1, position, l)} that, respectively, moves agent \sort{R1} in direction \sort{d} or reads the value \sort{l} from the position sensor of \sort{R1}. 
    The semantics of action \sort{move(R1, d)} consists of all singleton event of the form $\{\sort{move(R1, d)}\}$ with $d$ a constant direction value.
    We use the operation $\cdot:\sort{Action Action}\to \sort{Action}$ to construct a composite action $\mathtt{a1 \cdot a2}$ out of two actions \sort{a1} and \sort{a2}.

\paragraph{Agent}
\noindent
    An agent operationally specifies a component in rewriting logic. We give the specification of an agent as a rewrite theory, and provide the semantics of an agent as a component. An agent is a four tuple  $(\Lambda, \Omega, \Eq, \Rightarrow)$, each of whose elements we introduce as follow.

    The set of sorts $\Lambda$ contains the \sort{State} sort and the \sort{Action} sort, respectively for state and action terms.
    A pair of a state and a set of actions is called a configuration.
    The set of function symbols $\Omega$ contains $\phi : \sort{State} \times \sort{Action} \to \sort{State}$, that takes a pair of a state and an action term to produce a new state.
    The $(\Lambda,\Omega)$-equational theory $\Eq$ specifies the update function $\phi$.
    The set of equations that specify the function $\phi$ can make $\phi$ both a continuous or discrete function.

    The rule pattern in (\ref{eq:agent-rw}) updates a configuration with an empty set to a new configuration, i.e., 
     \begin{equation}
         \label{eq:agent-rw}
     (s,\emptyset) \Rightarrow (s', \acts)
 \end{equation}
     with $\acts$ a non-empty set of action terms, and $s'$ a new state.
     We call an agent \emph{productive} if, for any state $s: \sort{State}$, there exists a state $s'$ with $(s,\emptyset) \Rightarrow (s',\acts)$ and $\acts$ non empty set. Such agent may eventually do the idling action $\star$. 

    We give a semantics of an agent as a component by considering the limit application of the agent rewrite rules.
    We construct a TES transition system $\mathcal{T}_\mathcal{A} = (Q,E,\to)$ as an intermediate representation for agent $\mathcal{A} = (\Lambda, \Omega, \Eq,\Rightarrow)$.
    The set of states $Q = \sort{State}\times \mathbb{N}$ is the set of pairs of a state of $\mathcal{A}$ and a time-stamp natural number. We use the notation $[s,t]$ for states in $Q$ where $t \in \N$.
    The set of events $E$ is the union of all observables labeling the transition relation $\to\subseteq Q \times (\Po(E)\times \N) \times Q$, defined as the smallest set such that, for $t \in \N$:
    \begin{equation}
        \label{eqn:agent-rul}
    \cfrac{(s,\emptyset) \Rightarrow (s',\acts)\qquad a \in \acts \qquad \phi(s',a)=_\Eq s''}{[s,t] \xrightarrow{(a,t+1)} [s'',t+1]}
    \end{equation}

    An agent that performs a rewrite moves the global time from one unit forward. 
    All agents share the same time semantically, and we show some mechanisms at the system level to artificially run some agents \emph{faster} than others.

Let $\mathcal{A} = (\Lambda,\Omega,\Eq,\Rightarrow)$ be an agent initially in state $s_0 \in S$ at time $t_0\in\N$.
The finite, respectively infinite, component semantics of $\mathcal{A}$ is the component $\llbracket \mathcal{A}([s_0,t_0]) \rrbracket^* = (E,\Lfins{\mathcal{T}_\mathcal{A}, [s_0,t_0]})$, respectively the component $\llbracket \mathcal{A}([s_0,t_0]) \rrbracket = (E,\Linf{\mathcal{T}_\mathcal{A}, [s_0,t_0]})$, with $E = \bigcup_{ a \in \sort{Action}} a$.

\begin{lemma}[Closure]
    \label{lemma:closure}
    Let $\mathcal{A}$ be a productive agent initially in state $[s_0,t_0]$. Then  $\sem{A([s_0,t_0])}^* = \cl{\sem{A([s_0,t_0])}}$.
\end{lemma}

    Lemma~\ref{lemma:closure} gives a condition under which a step by step execution of the agent is sound with respect to generating prefixes of elements in the component semantics.
    More precisely, if an agent $\mathcal{A}$ is productive, Lemma~\ref{lemma:closure} ensures that finite sequences of rewrite rule applications generate finite sequences of observations each of which is a prefix of an element in the behavior of the component corresponding to $\mathcal{A}$.
    Alternatively, if $\mathcal{A}$ is not productive, a finite sequence of rule application may lead to a state for which no rule applies anymore. In such a case, there may not be any corresponding element in the agent component for which such finite sequence is a prefix. 

\paragraph{System} 
    A system gives an operational specification of a product of a set of components under $\Sigma = ([\kappa],\cup)$.
    The composability relation $\kappa$ is fixed to be symmetric, so that the product $\times_\Sigma$ is commutative. We define $[\kappa]$ co-inductively, as in~\cite{DBLP:journals/corr/abs-2110-02214, https://doi.org/10.48550/arxiv.2205.13008}.
    Formally, a system consists of a set of agents with additional sorts, operations, and rewrite rules.
    A system is a tuple $(\mathcal{A},\Lambda, \Omega, \Eq, \Rightarrow_S)$ where $\mathcal{A}$ is a set of agents. We use $(\Lambda_i, \Omega_i,\Eq_i,\Rightarrow_i)$ to refer to agent $\mc{A}_i \in \mc{A}$.

    The set of sorts $\Lambda$ contains a sort $\sort{Action} \in \Lambda$ which is a super sort of each sort $\sort{Action}_i$ for $\mc{A}_i \in \mc{A}$. 
    The set $\Omega$ contains the function symbol $\sort{comp}:\sort{Action}\times\sort{Action} \to \sort{Bool}$, which relates pairs of action terms.
    Given two actions \sort{a1,a2:Action}, $\mathtt{comp(a1,a2)}=\sort{True}$ when the two actions \sort{a1} and \sort{a2} are \emph{composable}.
    The set of equations $\Eq$ specifies the composability relation \sort{comp}. 
    First, we impose $\sort{comp}$ to be symmetric, i.e., for all actions \sort{a1,a2:Action}, $\mathtt{comp(a1,a2)}=\mathtt{comp(a2,a1)}$.
    Second, we assume that
    $\mathtt{comp(a1\cdot a2,\ a3)}$ and \sort{comp(a1, a2)} hold if and only if \sort{comp(a2, a3)} and $\mathtt{comp(a1,\ a2\cdot a3)}$ hold, for any actions \sort{a1, a2, a3} from disjoint agents.
    Given a set $\mathtt{actions}$ of actions, we use the notation $\mathtt{comp(actions)}$ 
    for the predicate that is \sort{True} if 
    all pairs of actions in \sort{actions} are composable, i.e., for all \sort{a1, a2} in \sort{actions}, $\mathtt{comp(a1,a2)}$ is \sort{True} and for all agent $\mc{A}_i$ such that there is no $\mathtt{a3:Action_i} \in \sort{actions}$, then $\mathtt{comp(a1,\star_i)}$ is \sort{True}.
    We call a set \sort{actions} of actions for which $\mathtt{comp(actions)}$ holds, a \emph{clique}.
The conditions for a set of actions to form a clique models the fact that each action in the clique is independent from agent $\mc{A}_i$ with no action in that clique (see Section~\ref{sec:general-framework} for an instance of \sort{comp}), and therefore composable with the silent action $\star_i$.%
\iftoggle{short}{
\begin{example}
    Suppose that action \sort{a1} from agent \sort{A1} reads the battery sensor value provided by agent \sort{A2} and the position sensor value provided by agent \sort{A3}.
    The composability relation \sort{comp} may impose that the read action \sort{a1} succeeds only if \sort{A2} and \sort{A3} do a write action \sort{a2} and \sort{a3}. In which case, the value read by \sort{A1} coincides with the value written by \sort{A2} and \sort{A3}, and the clique is, for instance, the set $\mathtt{\{read((battery, position), l, (x,y)), display(battery, l), display(position, (x,y))\}}$ where a robot does a \sort{read} action coincidentally with the \sort{display} action from its battery and the field on which it moves.
\end{example}
}{}
The relation \sort{comp} can be graphically modelled as an undirected graph relating actions, where a clique is a connected component. 

The rewrite rule pattern in (\ref{eq:sys-rw}) selects a set of actions, at most one from each agent, checks that the set of actions forms a clique with respect to \sort{comp}, and applies the update accordingly. For $\{k_1, ..., k_j\} \subseteq \{1,...,n\}$:
    \begin{equation} 
        \label{eq:sys-rw}
    \{(s_{k_1}, \acts_{k_1}), ..., (s_{k_j}, \acts_{k_j})\} \Rightarrow_S \{(\phi_{k_1}(s_{k_1}, a_{k_1}), \emptyset), ..., (\phi_{k_j}(s_{k_j}, a_{k_j}), \emptyset)\}     
\end{equation}
if $\sort{comp}(\bigcup_{i \in [1,j]} \{a_{k_i}\}))$.
As we show later, a system does not necessarily update all agents in lock steps, and an agent not doing an action may stay in the configuration $(s,\emptyset)$.
As multiple cliques may be possible, there is non-determinism at the system level. Different strategies may therefore choose different cliques as, for instance, taking the largest clique.

We define the transition system for $\mathcal{S} = (\mc{A},\Lambda, \Omega, \Eq, \Rightarrow_S)$ as the TES transition system $\tr{\mc{S}} = (Q,E,\to)$ with $Q = \sort{StateSet}\times\N$ the set of states, $E$ the union of all observables labeling the transition relation $\to \subseteq Q\times (\Po(E)\times\N) \times Q$, which is the smallest transition relation such that, for $\{k_1, ..., k_j\} \subseteq \{1,...,n\}$:
\begin{equation}
    \label{eqn:sys-rul}
\cfrac{
    \{(s_{k_i}, \acts_{k_i})\}_{i\in[1,j]} \Rightarrow_S \{(\phi_{k_i}(s_{k_i}, a_{k_i}),\emptyset)\}_{i\in[1,j]} \quad  \bigwedge_{i\in [1,j]} \phi_{k_i}(s_{k_i},a_{k_i}) =_{\Eq_i} s''_{k_i}
}
{[\{s_i\}_{i\in[1,n]},t] \xrightarrow{(\bigcup_{i\in[1,j]} a_{k_i},t+1)} [\{s_1,...,s''_{k_1}, ..., s''_{k_j}, ..., s_n\}, t+1]}
\end{equation}
for $t\in \N$ and where we use the notation $\{x_i\}_{i\in[1,n]}$ for the set $\{x_1,...,x_n\}$.
\begin{remark}
    The top left part of the rule is a rewrite transition at the system level. As defined earlier, the condition for such rewrite to apply is the formation of a clique by all of the actions in the update. The states and labels of the TES transition system (bottom of the rule) are sets of states and sets of labels from the TES transition system of every agent in the system.
\end{remark}

Let $\mathcal{A} = \{\mathcal{A}_1, ..., \mathcal{A}_n\}$ be a set of agents, and let $\mathcal{S} = (\mathcal{A}, \Lambda,\Omega,\Eq, \Rightarrow_S)$ be a system initially in state $\{(s_{0i},\emptyset)\}_{i\in[1,n]}$ at time $t_0$ such that, for all $i\in[1,n]$, $\mathcal{A}_i$ is initially in state $s_{0i}$ at time $t_0$. 
The finite, respectively infinite, semantics of initialized system $\mc{S}([s_0,t_0])$, is the component $\llbracket \mc{S}([s_0,t_0]) \rrbracket^* = (E,\Lfins{\tr{\mc{S}},[s_0,t_0]})$, respectively $\llbracket \mc{S}([s_0,t_0]) \rrbracket = (E,\Linf{\tr{\mc{S}},[s_0,t_0]})$, where $E = \bigcup_{i\in [1,n]}E_i$ with $E_i$ the set of events for the agent component $\llbracket \mc{A}([s_{0i},t_0]) \rrbracket$.

    Given a composability relation \sort{comp}, we define the interaction signature $\Sigma = ([\kappa_\sort{comp}],\cup)$, with $\kappa_\sort{comp}(E_1, E_2) \subseteq (\Po(E_1)\times\N)\times(\Po(E_2)\times\N)$ to be such that, for $\mathtt{ai:Action_i}$ and $\mathtt{aj:Action_j}$:
    \begin{itemize}
        \item if $\mathtt{comp(ai, aj)}$,  then $((a_i, n), (a_j,n)) \in \kappa_\sort{comp}(E_i,E_j)$ for all $n\in\N$, i.e., two composable actions occur at the same time;
        \item if $\mathtt{comp(ai, \star_j)}$, then $((a_i,n), (a,k)) \in \kappa_\sort{comp}(E_i,E_j)$ for all $(a,k) \in \Po(E_j)\times\N$ with $k\geq n$, i.e.,  $\mc{A}_j$ may have an action at arbitrary future time.
    \end{itemize}
    with $E_i$ the set of events of agent $\mc{A}_i$.
\begin{lemma}[Composability]
    \label{lem:prop}
    If $\mathtt{Action_i}\cap\mathtt{Action_j} = \emptyset$ for all disjoint agents $i$ and $j$, then 
    the product $\times_{([\kappa_\sort{comp}],\cup)}$ is commutative and associative.
\end{lemma}
\begin{theorem}[Compositional semantics]
    \label{thm:compositionality}
    Let $\mathcal{S} = (\mathcal{A}, \Lambda,\Omega,\Eq, \Rightarrow_S)$ be a system of $n$ agents with disjoint actions and $[\{s_{01},...,s_{0n}\},t_0]$ as initial state. We fix $\Sigma = ([\kappa_{\sort{comp}}],\cup)$. 
    Then,
    $\llbracket \mathcal{S}([s_0,t_0]) \rrbracket  = \times_\Sigma \{\llbracket \mathcal{A}_i([s_{0i},t_0]) \rrbracket\}_{i\in[1,n]}$ 
    and 
    $\llbracket \mathcal{S}([s_0,t_0]) \rrbracket^*  = \times_\Sigma \{\llbracket \mathcal{A}_i([s_{0i},t_0]) \rrbracket^*\}_{i\in[1,n]}$.
\end{theorem}

\section{Application}
\label{sec:implementation}

    We present the Maude implementation of the rewrite theories described in Section~\ref{sec:specification}. 
    We first describe our general framework as currently implemented in Maude, separating the agent modules, from the system module, and the composability relation. 
    The framework is instantiated for a system consisting of two robot agents, each interacting with a (shared) field and a (private) battery agent \iftoggle{short}{}{(more details can be found in~\ref{appendix:agent})}. 
    \iftoggle{short}{We show how to specify each agent equationally, and how to specify their interaction.}{}
    Finally, we run some analysis on the system using the Maude reachability search engine.
The implementation of the framework in Maude can be found in~\cite{cp-agents}.

\subsection{General framework}
\label{sec:general-framework}

\paragraph{Actions}
An action is a pair that contains the name of the action, and the set of agent identifiers on which the action applies.
An agent action is identified by the source agent identifier, and is a triple \texttt{(id, (a; ids))} where \texttt{id} is the agent doing the action with name \texttt{a} onto the set of agents \texttt{ids}, that we call resources of agent \texttt{id} for action named \texttt{a}.
\begin{lstlisting}
fmod ACTION is
    inc STRING . inc BOOL . inc SET{Id} . ...
    sort AName Action AgentAction .
    op (_;_) : AName Set{Id} -> Action [ctor] .
    op (_,_) : Id Action -> AgentAction [ctor] .
    op mta : -> AgentAction .
endfm
\end{lstlisting}

\paragraph{Agent}
    The \texttt{AGENT} module in Listing~\ref{lst:agent} defines the theories on which an agent relies, the \sort{Agent} sort, and operations that an agent instance must implement.
    The module is parametrized with a \sort{CSEMIRING} theory, that is used to rank actions of an agent. 
    Additionally, the \texttt{AGENT} includes modules that define state and action terms. 
    A term of sort \texttt{IdStates} is a pair of an identifier and a map of sort \texttt{MapKD}. \\
    A term of sort \sort{Agent} is a tuple \texttt{[id: C| state; ready?; softaction]}.
    The identifier \sort{id} is unique for each agent of the same class \sort{C}.
    The state \sort{state} of an agent is a map from keys to values. 
    For instance, the state of a robot has three keys, \texttt{position}, \texttt{energy}, and \texttt{lastAction}, with values in \texttt{Location}, \texttt{Status}, and \texttt{Bool}. 
    The flag \sort{ready?} is of sort \sort{Bool} and is \sort{True} when the agent has submitted a possibly empty list of actions, and \sort{False} otherwise. 
    The pending actions \sort{softaction} is a set of actions valued in the parametrized \texttt{CSEMIRING}. The use of a constraint semiring as a structure for action valuations enables various kinds of reasoning about preferences at the agent and system levels.
We use the two operations of the csemiring, sum $+$ and product $\times$,  as respectively modeling the choice and the compromise of two alternatives. See~\cite{wirsing-etal-2007, talcott-arbab-yadav-15wirsing-fest, kappe-etal-2019scp-sca} for more details.\\
    An agent instance implements four operations: \texttt{computeActions},  \texttt{getOutput},  \texttt{getPostState}, and \texttt{internalUpdate}.
    The operation \texttt{computeActions}, given a \sort{state:MapKD} of agent \sort{id} of class \sort{C}, returns a set of valued actions in the parametrized \sort{CSEMIRING}.
    The operation \texttt{internalUpdate}, given a \sort{state:MapKD} of agent \sort{id} of class \sort{C}, returns a new state \sort{state':MapKD}.
    For instance, an agent may record in its state, as an internal update, the outcome of \texttt{computeActions} and change the value that the key \texttt{lastAction} maps to.
    The \texttt{getOutput} operation, given an action name \sort{a:Name} from agent identified by \sort{id2} applied to an agent \sort{id} of class \sort{C} in a state \sort{state}, returns a collection of outputs \sort{outputs = getOutput(id, C, id2, an, state)}.
    The outputs generated by \sort{getOutput} are of sort \sort{MapKD} and therefore structured as a mapping from keys to values. For instance, the output of the action named \sort{read} applied on a field agent has a key \sort{position} that maps to the position value of the agent doing the \sort{read} action.
    The operation \texttt{getPostState}, given an action name \sort{a:AName} with inputs \sort{input:IdStates} from agent identified by \sort{id2} applied on an agent \sort{id1} of class \sort{C} in a state \sort{state}, returns a new state \sort{state' = getPostState(id1, C, id2, an, input, state)}.
    The input \sort{input:IdStates} is a collection of key to value mappings that results from collecting the outputs, i.e., with \sort{getOutput}, of an action \sort{(id, an, ids)} on all its resources in \sort{ids}.

    {
\begin{lstlisting}[basicstyle=\footnotesize\ttfamily, caption={Extract from the \sort{AGENT} Maude module.}\label{lst:agent}]
fmod AGENT{X :: CSEMIRING} is
 inc IDSTATE .  inc ACTION .
 sort Agent .
 op [_:_|_;_;_] : Id Class MapKD Bool X$Elt -> Agent [ctor].
 op computeActions : Id Class MapKD -> X$Elt .
 op internalUpdate : Id Class MapKD -> MapKD .
 op getPostState : Id Class Id AName IdStates MapKD -> MapKD 
 op getOutput : Id Class Id AName MapKD -> MapKD .
endfm
\end{lstlisting}
}

The agent's dynamics are given by the rewrite rule in Listing~\ref{lst:rw-agent}, that updates the pending action to select one atomic action from the set of valued actions:
\begin{lstlisting}[caption={Conditional rewrite rule applying on agent terms.}\label{lst:rw-agent}]
crl[agent] : [sys [id : ac | state  ; false ; null]] => 
      [sys [id : ac | state' ; true  ; softaction]] 
    if softaction + sactions := computeActions(id, ac, state) 
       /\ state' := internalUpdate(id, ac, state) .
\end{lstlisting}
The rewrite rule in Listing~\ref{lst:rw-agent} implements the abstract rule of Equation~\ref{eqn:agent-rul}.
After application of the rewrite rule, the \sort{ready?} flag of the agent is set to \texttt{True}. The agent may, as well, perform an internal update independent of the success of the selected action.

\paragraph{System}
    The \sort{SYSTEM} module in Listing~\ref{lst:sys} defines the sorts and operations that apply on a set of agents.
    The sort \sort{Sys} contains set of \sort{Agent} terms, and the term \sort{Global} designates top level terms on which the system rewrite rule applies (as shown in Listing~\ref{lst:rwl-sys}).
    The \sort{SYSTEM} module includes the \sort{Agent} theory parametrized with a fixed semiring \sort{ASemiring}.
    The theory \sort{ASemiring} defines valued actions as pairs of an action and a semiring value. While we assume that all agents share the same valuation structure, we can also define systems in which such a preference structure differs for each agent.
    The \sort{SYSTEM} module defines three operations: \sort{outputFromAction}, \sort{updateSystemFromAction}, and \sort{updateSystem}.
    The operation \sort{outputFromAction} returns, given an agent action \sort{(id, (an, ids))} applied on a system \sort{sys}, a collection of identified outputs \sort{idOutputs = outputFromAction((id, (an, ids)), sys)} given by the union of \sort{getOutput} from all agents in \sort{ids}.
    The operation \sort{updatedSystemFromAction} returns, given an agent action \sort{(id, (an, ids))} applied on a system \sort{sys}, an updated system \sort{sys' = updatedSystemFromAction((id, (an, ids)), sys)}. 
    The updated system may raise an error if the action is not allowed by some of the resource agents in \sort{ids} (see the battery-field-robot example in~\ref{appendix:agent}).
    The updated system, otherwise, updates \emph{synchronously} all agents with identifiers in \sort{ids} by using the \sort{getPostState} operation.
    The operation \sort{updateSystem} returns, given a list of agent actions \sort{agentActions} and a system term \sort{sys}, a new system \sort{updateSystem(sys, agentActions)} that performs a sequential update of \sort{sys} with every action in \sort{agentActions} using \sort{updatedSystemFromAction}.
    The list \sort{agentActions} ends with a delimiter action \sort{end} performed on every agent, which may trigger an error if some expected action does not occur (see \sort{PROTOCOL} in~\ref{appendix:agent}).

    \vspace{-1em}
\begin{lstlisting}[caption={Extract from the \sort{SYSTEM} Maude module.}\label{lst:sys}]
fmod SYS is
  inc AGENT{ASemiring} . sort Sys  Global . 
  subsort Agent < Sys . op [_] : Sys -> Global [ctor] .
  op __ :  Sys Sys -> Sys [ctor assoc comm id: mt] .   ...
  op outputFromAction : AgentAction Sys -> IdStates .
  op updatedSystemFromAction : AgentAction Sys -> Sys .
  op updateSystem : Sys List{AgentAction} -> Sys .
endfm
\end{lstlisting}
    The rewrite rule in Listing~\ref{lst:rwl-sys} applies on terms of sort \sort{Global} and updates each agent of the system synchronously, given that their actions are composable.
    The rewrite rule in Listing~\ref{lst:rwl-sys} implements the abstract rule of Equation~\ref{eqn:sys-rul}.
    The rewrite rule is conditional on essentially two predicates: \texttt{agentsReady?} and \texttt{kbestActions}.
    The predicate \sort{agentsReady?} is \sort{True} if every agent has its \sort{ready?} flag set to \sort{True}, i.e., the agent rewrite rule has already been applied.
    The operation \sort{kbestActions} returns a ranked set of cliques (i.e., composable lists of actions), each paired with the updated system.
    The element of the ranked set are lists of actions containing at most one action for each agent, and paired with the system resulting from the application of \sort{updateSystem}. If the updated system has reached a \sort{notAllowed} state, then the list of actions is not composable and is discarded.
    The operations \sort{getSysSoftActions} and \sort{buildComposite} form the set of lists of composite actions, from the agent's set of ranked actions, by composing actions and joining their preferences.

\begin{lstlisting}[caption={Conditional rewrite rule applying on system terms.}\label{lst:rwl-sys}]
crl[transition] : [sys]  => [sys']
 if agentsReady?(sys)  /\ saAtom := getSysSoftActions(sys) /\ 
  saComp := buildComposite(saAtom , sizeOfSum(saAtom)) /\
  p(actseq, sys') ; actseqs := kbestActions(saComp, k, sys) .
\end{lstlisting}
\paragraph{Composability relation}
The term \sort{saComp} defines a set of valued lists of actions. 
Each element of \sort{saComp} possibly defines a clique. 
The operation \sort{kbestActions} specifies which, from the set \sort{saComp}, are cliques. We describe below the implementation of \sort{kbestActions}, given the structure of action terms.\\
An action is a triple \sort{(id, (an, ids))}, where \sort{id} is the identifier of the agent performing the action \sort{an} on resource agents \sort{ids}.
Each resource agent in \sort{ids} reacts to the action \sort{(id, (an, ids))} by producing an output \sort{(id', an, O)} (i.e., the result of \sort{getOutput}). Therefore, $\mathtt{comp((id, (an, ids)), a_i)}$ holds, with $\mathtt{a_i:Action_i}$ and $i \in \sort{ids}$, only if $\mathtt{a_i}$ is a list that contains an output \sort{(i, an, O)}, i.e., an output to the action.
If one of the resources outputs the value \sort{(i,notAllowed(an))}, the set is discarded as the actions are not pairwise composable. Conceptually, there are as many action names \sort{an} as possible outputs from the resources, and the system rule~(\ref{eqn:agent-rul}) selects the clique for which the action name and the outputs have the same value. In practice, the list of outputs from the resources get passed to the agent performing the action. 

\iftoggle{short}{
\subsection{Instances}

    Section~\ref{sec:general-framework} introduces the signature for an agent module, and the rewrite rules for an agent and a system.
    The instance of an agent module provides an equational theory that implements each operation, namely \texttt{computeActions}, \texttt{internalUpdate}, \texttt{getOutput}, and \texttt{getPostState}.
    Each instance comes with an interface, called \texttt{AGENT-INTERFACE} in which the action names for \texttt{AGENT} are constructed.
    For instance, the interface for the robot agent called \texttt{TROLL}  contains the constructors for action names \texttt{move: direction -> Action} and \texttt{read: sensorName -> Action}.
    An agent that interacts with another agent must therefore include the interface module of that agent.
    We also assume that each agent shares the same preference structure, which we call action semiring (written \texttt{ASemiring}).
    The action semiring consists of an action paired with a natural number preference value.
     To illustrate the use of our framework to simulate and verify cyber-physical systems, we present an agent specification for three components: a \texttt{FIELD}, a \texttt{TROLL}, and a \texttt{BATTERY}. 
     A \texttt{FIELD} component interacts with the \texttt{TROLL} component by reacting to its move action, and its sensor reading.
     As shown in Listing~\ref{lst:field} the \texttt{FIELD} agent has no actions, but reacts to the move action of the \texttt{TROLL} agent by updating its state and changing the agent's location. 
     Currently, the update is discrete, but more sophisticated updates can be defined (e.g., changing the mode of a function recording the trajectory of the \sort{TROLL} agent).
     In the case where the state of the \sort{FIELD} agent forbids the \sort{TROLL} agent's move, the \sort{FIELD} agent enters in a disallowed state marked as \texttt{notAllowed(an)}, with \texttt{an} as the action name.
     The \sort{FIELD} responds to the read sensor action by returning the current location of the \sort{TROLL} agent as an output.
}{}

\iftoggle{short}{
     A \texttt{TROLL} agent reacts to no other agent actions, and therefore does not include any agent interface.
     However, the \texttt{TROLL} agent returns a ranked set of actions given its state with the \texttt{computeActions} operation. The expression may contain more than one action, with different weights. The weights of the action may depend on the internal goal that the agent set to itself, as for instance reaching a location on the field.
     The \texttt{TROLL} agent specifies how it reacts to, e.g., the sensor value input from the field, by updating the corresponding field in its state, i.e., with \texttt{getSensorValues}.
}{}

\iftoggle{short}{
     A \texttt{BATTERY} agent does not act on any other agent, as the \texttt{FIELD}, but reacts to the \texttt{TROLL} agent actions.
     Each \texttt{move} action triggers in the \sort{BATTERY} agent a change of state that decreases its energy level. As well, each \texttt{charge} action changes the \sort{BATTERY} agent state to increase its energy level.
     Similarly to the field, in the case where the state of the battery agent has $0$ energy, the battery enters a disallowed state marked as \texttt{notAllowed(an)}, with \texttt{an} as the action name.
     A sensor reading by the \texttt{TROLL} agent triggers an output from the \sort{BATTERY} agent with the current energy level.
}{}

\iftoggle{short}{
\textcolor{blue}{
\paragraph{Composability relation}
The \sort{TROLL}, \sort{FIELD}, and \sort{BATTERY} modules specify the state space and transition functions for, respectively, a \sort{TROLL}, \sort{FIELD}, and \sort{BATTERY} agent. A system consisting of a set of instances of such agents would need a composability relation to relate actions from each agent.\\
More precisely, we give the possible \emph{cliques} of a system consisting of two \sort{TROLL} agents with identifiers  \sort{id(0), id(1):TROLL}, one \sort{field:FIELD} agent, and two \sort{BATTERY} agents \sort{bat(0),bat(1):BATTERY}. \\
The action of agent \sort{id(0)} composes with outputs of its corresponding battery \sort{bat(0)} and of the shared \sort{field} agent.\\
For instance, a move action of the \sort{id(0)} agent is of the form \sort{(id(0), (move(d), \{bat(0), field\}))}, where \sort{d} is a direction for the move, and composes with outputs of the battery and field, both notifying that the move is possible.\\
Alternatively, a read action of the \sort{id(0)} agent is of the form \sort{(id(0), (read, \{bat(0), field\}))} and composes with outputs of the battery and field, each giving the battery level and the location of agent \sort{id(0)}.
}

\paragraph{System} 
The agents defined above are instantiated within the same system to study their interactions.
We consider a system containing two \texttt{TROLL} agents, with identifiers \texttt{id(0)} and \texttt{id(1)}, paired with two \sort{BATTERY} agents with identifier \texttt{bat(0)} and \texttt{bat(1)}, and sharing the same \texttt{FIELD} resource. 
\textcolor{blue}{
The goal for each agent is to reach the initial location of the other agent. If both agent follow the shortest path to their goal location, there is an instant for which the two agents need to swap their positions.
The crossing can lead to a livelock, where agents move symmetrically until the energy of the batteries runs out.
}
}{}

\subsection{Analysis in Maude}
We analyze in Maude two scenarios. In one, each robot has as strategy to take the shortest path to reach its goal. As a consequence, a robot reads its position, computes the shortest path, and submits a set of optimal actions. A robot can sense an obstacle on its direct next location, which then allows for sub-optimal lateral moves (e.g., if the obstacle is in the direct next position in the West direction, the robot may go either North or South). 
In the other scenario, we add a protocol that swaps the two robots if robot \texttt{id(0)} is on the direct next location on the west of robot \texttt{id(1)}. The swapping is a sequence of moves that ends in an exchange of positions of the two robots. 
\iftoggle{short}{}{See \ref{appendix:agent} for details on the \sort{TROLL}, \sort{FIELD}, \sort{BATTERY}, and \sort{PROTOCOL} agents specified in Maude, and for the specification of the \sort{init} term for both scenarios.}

In the two scenarios, we analyze the behavior of the resulting system with two queries.
The first query asks if the system can reach a state in which the energy level of the two batteries is $0$, which means that its robot can no longer move: 
\vspace{-0.3em}
\begin{lstlisting}
search [1] init =>* [sys::Sys  
       [ bat(1) : Battery | k(level) |-> 0 ; true ; null],
       [ bat(2) : Battery | k(level) |-> 0 ; true ; null]] .
\end{lstlisting}
 The second query asks if the system can reach a state in which the two robots successfully reached their goals, and end in the expected locations:
\vspace{-0.3em}
\begin{lstlisting}
search [1] init =>* [sys::Sys  [ field : Field | k(( 5 ; 5 ))
  |-> d(id(0)), k(( 0 ; 5 )) |-> d(id(1)) ; true ; null]] .
\end{lstlisting}

  As a result, when the protocol is absent, the two robots can enter in a livelock behavior and eventually fail with an empty battery: 
  \begin{lstlisting}
Solution 1 (state 80)                                                                                            
states: 81  rw: 223566 in 73ms cpu (74ms real) (3053554 rw/s)
  \end{lstlisting}

Alternatively, when the protocol is used, the livelock is removed using exogenous coordination.
  The two robots therefore successfully reach their end locations, and stop before running out of battery: 
  \begin{lstlisting}
No solution. states: 102 
rewrites: 720235 in 146ms cpu (145ms real) (4920041 rw/s)
  \end{lstlisting}

  In both cases, the second query succeeds, as there exists a path for both scenarios where the two robots reach their end goal locations.
  The results can be reproduced by downloading the archive at~\cite{cp-agents}.
\section{Related work}

\vspace{-0.2em}
\paragraph{Real-time Maude}
Real-Time Maude is implemented in Maude as an extension of Full Maude~\cite{DBLP:conf/wrla/Olveczky14}, and is used in applications such as in~\cite{DBLP:conf/cav/LeeKBO21}.
    There are two ways to interpret a real-time rewrite theory, called the pointwise semantics and the continuous semantics.
    Our approach to model time is similar to the pointwise semantics for real-time Maude, as we fix a global time stamp interval before execution. The addition of a composability relation, that may discard actions to occur within the same rewrite step, differs from the real-time Maude framework.

\vspace{-0.2em}
\paragraph{Models based on rewriting logic} 
In [21], the modeling of cyber-physical systems from an actor perspective is discussed. The notion of event comes as a central concept to model interaction between agents.
Softagents \cite{talcott-arbab-yadav-15wirsing-fest} is a framework for specifying and analyzing adaptive cyber-physical systems implemented in Maude. It has been used to analyze systems such as vehicle platooning \cite{dantas20vnc} and drone surveillance \cite{mason17cosim}.
In Softagents agents interact by sharing knowledge and resources implemented as part of the system timestep rule.

    Softagents only considers compatibility in the sense of reachability of desired
or undesired states.   Our approach provides more structure enabling static
analysis.
Our framework allows, for instance, to consider compatibility of a robot with a battery (i.e., changing the battery specification without altering other agents in the system), and coordination of two robots with an exogenous protocol, itself specified as an agent.

\vspace{-0.2em}
\paragraph{Algebra, co-algebra}
The algebra of components described in this paper is an extension of~\cite{DBLP:journals/corr/abs-2110-02214}.
Algebra of communicating processes~\cite{DBLP:books/daglib/0000497} (ACP) achieves similar objectives as decoupling processes from their interaction. For instance, the encapsulation operator in process algebra is a unary operator that restricts which actions may occur, i.e., $\delta_H(t \parallel s)$ prevents $t$ and $s$ to perform actions in $H$. Moreover, composition of actions is expressed using communication functions, i.e., $\gamma(a,b)=c$ means that actions $a$ and $b$, if performed together, form the new action $c$.
Different types of coordination over communicating processes are studied in~\cite{BERGSTRA1984109}.

\vspace{-0.2em}
\paragraph{Discrete Event Systems}
Our work represents both cyber and physical aspects of systems in a unified model of discrete event systems~\cite{Nivat1982, DBLP:conf/birthday/Arbab11}.
In~\cite{doi:10.1146/annurev-control-053018-023659}, the author lists the current challenges in modelling cyber-physical systems in such a way. 
The author points to the problem of modular control, where even though two modules run without problems in isolation, the same two modules may block when they are used in conjunction. 
In~\cite{DBLP:journals/tac/SampathLT98}, the authors present procedures to synthesize supervisors that control a set of interacting processes and, in the case of failure, report a diagnosis. 
An application for large scale controller synthesis is given in~\cite{MOORMANN2021104902}.
Our framework allows for experiments on modular control, by adding an agent controller among the set of agents to be controlled. The implementation in Maude enables the search of, for instance, blocking configurations.

\vspace{-0.8em}
\section{Conclusion}
We give an operational specification of the algebra of components defined in~\cite{DBLP:journals/corr/abs-2110-02214}.
An agent specifies a component as a rewrite theory, and a system specifies a product of components as a set of rewrite theories extended with a composability relation.
We show compositionality, i.e., that the system specifies a component that equals to the product, under a suitable interaction signature, of components specified by each agent.

We present an implementation of our framework in Maude, and instantiate a set of components to model two energy sensitive robots roaming on a shared field. 
We analyze the behavior of the resulting system before and after coordination with a protocol, and show how the protocol can prevent livelock behavior.

The modularity of our operational framework and the interpretation of agents as components in interaction add structure to the design of cyber-physical systems. The structure can therefore be exploited to reason about more general properties of CPSs, such as compatibility, sample period synthesis, etc.

\paragraph{Acknowledgement}
Talcott was partially supported by the U. S. Office of Naval Research under award numbers N00014-15-1-2202 and N00014-20-1-2644, and NRL grant N0017317-1-G002.
Arbab was partially supported by the U. S. Office of Naval Research under award number N00014-20-1-2644.

\bibliographystyle{plain}
\bibliography{references}
\newpage
\appendix
\section{TES transition system}
\label{appendix:TTS}

The behavior of a component as in Definition~\ref{def:component} is a set of TESs. 
We give an specification of such behavior using a labelled transition system.
\begin{definition}[TES transition system]
A TES transition system is a triple $(Q,E,\rightarrow)$ where $Q$ is a set of states, $E$ is a set of events, and $\rightarrow \subseteq Q \times (\Po(E) \times \Rp) \times Q$ is a labeled transition relation, where labels are observations.
\hfill$\triangle$
\end{definition}

We present two different ways to give a semantics to a TES transition system: inductive and co-inductive. 
Both definitions give the same behavior, as shown in Theorem $1$ in~\cite{https://doi.org/10.48550/arxiv.2205.13008}.

\paragraph{Semantics 1 (runs).}
Let $T = (Q,E,\rightarrow)$ be a TES transition system.
Given $s \in (\Po(E)\times \Rp)^n$, we write $q \xrightarrow{s} p$ for the sequence of transitions $q \xrightarrow{s(0)} q_1 \xrightarrow{s(1)} q_2\ ... \xrightarrow{s(n)} p$. We use $\to^*$ and $\to^\omega$ to denote, respectively, the set of finite and infinite sequences of consecutive transitions in $\to$.
Then, finite sequences of observables form the set $\Lfin{T,q} = \{ \sigma \in \TES{E} \mid q \xrightarrow{s} q', \exists n. s = \sigma[n] \land \sigma^{(n)} \in \TES{\emptyset} \}$ and infinite ones, the set $\Linf{T,q} = \{ \sigma \in \TES{E} \mid \forall n. \sigma[n] \in \Lfin{T,q} \}$ where, as introduced in Section~\ref{sec:component}, $\sigma[n]$ is the prefix of size $n$ of $\sigma$.
The semantics of such a TES transition system $T = (Q,E,\rightarrow)$, starting in a state $q \in Q$, is the component $C_T(q) = (E, \Linf{T,q})$. 

\paragraph{Semantics 2 (greatest post fixed point)}
Alternatively, the semantics of a TES transition system is the greatest post fixed point of a function over sets of TESs paired with a state. 
For a TES transition system $T = (Q, E, \to)$, let $\mathcal{R} \subseteq \TES{E} \times Q$. We introduce $\phi_{T} : \Po(\TES{E}\times Q) \rightarrow \Po(\TES{E}\times Q)$
as the function:
$$
\begin{array}{rl}
    \phi_T(\Rel) = \{ (\tau,q) \mid   & \exists p \in Q,\ q\xrightarrow{\tau(0)}p \land (\tau',p) \in \Rel\}
\end{array}
$$

The product of two components is parametrized by a composability relation and a composition function and syntactically constructs the product of two TES transition systems.
\begin{definition}[Product]
The product of two TES transition systems $T_1 = (Q_1, E_1, \to_1)$ and $T_2 = (Q_2, E_2, \to_2)$ under the constraint $\kappa$ is the TES transition system $T_1 \times_\kappa T_2 = (Q_1 \times Q_2, E_1 \cup E_2, \to)$ such that:
\[
\cfrac{
    q_1 \xrightarrow{(O_1, t_1)}_1 q_1' \quad q_2 \xrightarrow{(O_2, t_2)}_2 q_2' \quad ((O_1, t_1), (O_2, t_2)) \in \kappa(E_1, E_2) \quad t_1 < t_2
}{
(q_1,q_2) \xrightarrow{(O_1, t_1)} (q_1',q_2)
}
\]
\[
\cfrac{
    q_1 \xrightarrow{(O_1, t_1)}_1 q_1' \quad q_2 \xrightarrow{(O_2, t_2)}_2 q_2' \quad ((O_1, t_1), (O_2, t_2)) \in \kappa(E_1, E_2) \quad t_2 < t_1
}{
(q_1,q_2) \xrightarrow{(O_2, t_2)} (q_1,q_2')
}
\]
\[
\cfrac{
    q_1 \xrightarrow{(O_1, t_1)}_1 q_1' \quad q_2 \xrightarrow{(O_2, t_2)}_2 q_2' \quad ((O_1, t_1), (O_2, t_2)) \in \kappa(E_1, E_2) \quad t_1 = t_2
}{
(q_1,q_2) \xrightarrow{(O_1\cup O_2, t_1)} (q_1',q_2')
}
\]
\hfill$\triangle$
\end{definition}

Observe that the product is defined on pairs of transitions, which implies that if $T_1$ or $T_2$ has a state without outgoing transition, then the product has no outgoing transitions from that state. The reciprocal is, however, not true in general.

Theorem $1$ in~\cite{https://doi.org/10.48550/arxiv.2205.13008} states that the product of TES transition systems denotes
(given a state) the set of TESs that corresponds to the product of the
corresponding components (in their respective states).
Then, the product that we define on TES transition systems does not add nor
remove behaviors with respect to the product on their respective components.

\section{Proofs}
\label{appendix:proof}

\begin{proof}
    Given that $\sem{A([s_0,t_0])}^* = (E,\Lfins{\mathcal{T}_\mathcal{A}, [s_0,t_0]})$ and $\cl{\sem{A([s_0,t_0])}} = (E,\Linf{\mathcal{T}_\mathcal{A}, [s_0,t_0]})$, we have to show that $\Lfins{\mathcal{T}_\mathcal{A}, [s_0,t_0]} = \cl{\Linf{\mathcal{T}_\mathcal{A}, [s_0,t_0]}}$.
    \begin{align*}
        \cl{\Linf{\mathcal{T}_\mathcal{A}, [s_0,t_0]}} & = \{ s\tau \in \TES{E}{}\mid \tau \in \TES{\emptyset}{}\ \ \it{and}\\
                                                       &\qquad\qquad\qquad\quad\qquad\ \exists \sigma.\exists i. \sigma \in \Linf{\mathcal{T}_\mathcal{A}, [s_0,t_0]} \land \sigma[i] = s\} \\
                                                       & = \{ s\tau \in \TES{E}{}\mid \tau \in \TES{\emptyset}{}\ \ \it{and}\\
                                                       &\qquad\qquad\qquad\quad\qquad\ \exists \sigma.\exists i. \forall n. \sigma[n] \in \Lfin{\mathcal{T}_\mathcal{A}, [s_0,t_0]} \land \sigma[i] = s\} \\
                                                       &\subseteq \Lfins{\mathcal{T}_\mathcal{A}, [s_0,t_0]}
    \end{align*}

    The other direction comes from the assumption that $\mc{A}$ is productive. Then, every reachable state in $\tr{A}$ has an outgoing transition and therefore every finite sequence of transition is a prefix of an infinite sequence.
    Thus, $\Lfins{\mathcal{T}_\mathcal{A}, [s_0,t_0]} \subseteq \cl{\Linf{\mathcal{T}_\mathcal{A}, [s_0,t_0]}}$.
    \qed
\end{proof}

\begin{proof}[Sketch - Lemma~\ref{lem:prop}]
    We abbreviate $\kappa_\sort{comp}$ to $\kappa$, and use $\Sigma = ([\kappa],\cup)$.
    We know that:
    \begin{enumerate}
        \item for all actions \sort{a1} and \sort{a2}, \sort{comp(a1,a2) = comp(a2,a1)};
        \item for all $\mathtt{a1:Action_1}$, $\mathtt{a2:Action_2}$, and $\mathtt{a3:Action_3}$, \sort{comp(a1, a2)} and $\mathtt{comp(a1\cdot a2,\ a3)}$ if and only if \sort{comp(a2, a3)} and $\mathtt{comp(a1,\ a2\cdot a3)}$.
    \end{enumerate}
    Item 1 implies symmetry of $\kappa$ and commutativity of $\times_\Sigma$.

    We show that, for three observations $(a_1,n)$, $(a_2, k)$, and $(a_3,l)$:
    \begin{align*}
    &((a_1,n), (a_2, k)) \in \kappa(E_1,E_2)\land ((a_1,n)+(a_2, k), (a_3,l)) \in \kappa(E_1\cup E_2,E_3)\\
 \iff &((a_2,k), (a_3, l)) \in \kappa(E_2,E_3)\land ((a_1,n),(a_2, k)+(a_3,l)) \in \kappa(E_1,E_2\cup E_3)
 \end{align*}
    where $((a,u),(b,v)) = (a\cup b, u)$ if $u = v$, $(a,u)$ if $u<v$, and $(b,v)$ otherwise.

    Suppose that $n=k=l$.
    Then, 
    \begin{align*}
         & ((a_1,n), (a_2, n)) \in \kappa(E_1,E_2)\land ((a_1\cup a_2, n), (a_3,n)) \in \kappa(E_1\cup E_2,E_3) \\
        \iff &\ \sort{comp(a1,a2)} \land  \mathtt{comp(a1\cdot a2, a3)} \\
        \iff &\ \mathtt{comp(a2, a3)} \land  \mathtt{comp(a1,a2\cdot a3)} \\
        \iff & ((a_2,n), (a_3, n)) \in \kappa(E_2,E_3)\land ((a_1, n), (a_2\cup a_3,n)) \in \kappa(E_1,E_2\cup E_3) 
    \end{align*}
    The second equivalence follows from $E_1$ and $E_2$ being disjoint.

    Suppose that $n< k$, then $((a_1,n), (a_2,k)) \in \kappa(E_1,E_2)$ if and only if $((a_1,n), (\emptyset,n)) \in \kappa(E_1,E_2)$, by definition of $\kappa$.
    Thus, for $n=l<k$, we have:
    \begin{align*}
         & ((a_1,n), (a_2, k)) \in \kappa(E_1,E_2)\land ((a_1, n), (a_3,n)) \in \kappa(E_1\cup E_2,E_3) \\
        \iff& ((a_1,n), (\emptyset, n)) \in \kappa(E_1,E_2)\land ((a_1, n), (a_3,n)) \in \kappa(E_1\cup E_2,E_3) \\
        \iff &\ \mathtt{comp(a1,\star_2)} \land  \mathtt{comp(a1\cdot \star_2, a3)} \\
        \iff &\ \mathtt{comp(\star_2, a3)} \land  \mathtt{comp(a1,\star_2\cdot a3)} \\
        \iff & ((a_2,k), (a_3, n)) \in \kappa(E_2,E_3)\land ((a_1, n), (a_3,n)) \in \kappa(E_1,E_2\cup E_3) 
    \end{align*}
    Similar reasoning apply when $n\not = l$ or $l \not = k$.

    We can conclude that $\kappa_\sort{comp}$ satisfies the condition of Lemma 7 in~\cite{DBLP:journals/corr/abs-2110-02214}, and $\times_{([\kappa_\sort{comp}],\cup)}$ is commutative and associative.
\end{proof}

\begin{proof}[Sketch - Theorem~\ref{thm:compositionality}]
    The proof uses the result of Lemma~\ref{lem:prop} that $\times_{([\kappa_\sort{comp}],\cup)}$ is associative and commutative. 
    Then, we give an inductive proof that $\llbracket \mathcal{S}([s_0,t_0]) \rrbracket  = \times_\Sigma \{\llbracket \mathcal{A}_i([s_{0i},t_0]) \rrbracket\}_{i\in[1,n]}$.
    We fix $\mc{S} = (\{\mc{A}_1, ..., \mc{A}_n\},\Lambda,\Omega, \Eq,\Rightarrow_S)$ and $\mc{A}_{n+1} = (\Lambda_{n+1},\Omega_{n+1},\Eq_{n+1},\Rightarrow_{n+1})$, such that \sort{comp} in $\Omega$ relates action of agents in $\{\mc{A}_1,..., \mc{A}_{n+1}\}$. Let $\mc{S}' = (\{\mc{A}_1, ..., \mc{A}_n,\mc{A}_{n+1}\},\Lambda, \Omega,\Rightarrow_S)$. 

    We show that $\tr{\mc{S}} \times_\kappa \tr{\mc{A}_{n+1}} = (Q,E,\to)$ and $\tr{\mc{S}'} = (Q',E',\to')$ are bisimilar, which consists in the existence of a relation $\mathcal{R} \subseteq Q \times Q'$ such that, for all $(q,r) \in \Rel$:
    \begin{enumerate}
        \item $\forall q' \in Q$ with $q \xrightarrow{(O,t)} q'$, there exists $r' \in Q'$ with $r \xrightarrow{(O,t)} r'$; and
        \item $\forall r' \in Q'$ with $r \xrightarrow{(O,t)} r'$, there exists $q' \in Q$ with $q \xrightarrow{(O,t)} q'$.
    \end{enumerate}

    First, 
    we define an equivalence relation $\sim$ on states in $Q$ as $([s_\mc{S},t], [s_\mc{A},t']) \sim ([s_\mc{S},\max(t,t')], [s_\mc{A},\max(t',t)])$. Then, we define the  set of states $Q_\sim$ such that $([s_\mc{S},\max(t,t')], [s_\mc{A},\max(t',t)]) \in Q_\sim$ if and only if $([s_\mc{S},t'], [s_\mc{A},t]) \in Q$  or $([s_\mc{S},t'], [s_\mc{A},t]) \in Q$.
    We show that the TES transition system $\tr{\mc{S}}\times_\kappa\tr{\mc{A}_{n+1}}$ projected to states in $Q_\sim$ is bisimilar to $\tr{\mc{S}}\times_\kappa\tr{\mc{A}_{n+1}}$. The reason is that the transition rules in $\tr{\mc{S}}$ and $\tr{\mc{A}_{n+1}}$ universally quantify over time $t\in\N$, which allows arbitrary positive translation in time. 
    As a consequence, states in $Q'$ can be embedded in states in $Q$.

    We now prove $1$ and $2$ by showing that $([s_1',t]) \xrightarrow{(O,t')}_{'} ([s_2',t'])$ if and only if $([s_1,t], [q_1,t]) \xrightarrow{(O,t')} ([s_2,\max(t',t)], [q_2,\max(t,t')])$ where $s_1$ and $s_2$ are states in $\tr{\mc{S}}$, $q_1$ and $q_2$ are states in $\tr{\mc{A}_{n+1}}$ and $s_1'$ and $s_2'$ are states in $\tr{\mc{S}'}$.
    We split cases on whether the observation comes from $\mc{S}$, from $\mc{A}$, or is a joint observation.
    We use the equational theory of the system to prove the result. 
$\triangle$
\end{proof}

\iftoggle{short}{}{
\section{Additional examples}
\subsection{A system of cyber-physical agents: an example}
\label{appendix:example}

This section illustrates our approach on an intuitive and simple cyber-physical system consisting of two robots roaming on a shared field.
A robot exhibits some cyber aspects, as it takes discrete actions based on its readings. Every robot interacts, as well, with a shared physical resource as it moves around. 
The field models the continuous response of each action (e.g., read or move) performed by a robot. 
A question that will motivate the paper is: given a strategy for both robots (i.e., sequence of moves based on their readings), will both robots, sharing the same physical resource, achieve their goals? If not, can the two robots, without changing their policy, be externally coordinated towards their goals?

In this paper, we specify components in a rewriting framework in order to simulate and analyze their behavior.
In this framework, an \emph{agent}, e.g., a robot or a field, specifies a component as a rewriting theory.
A \emph{system} is a set of agents that run concurrently.
The equational theory of an agent defines how the agent states are updated, and may exhibit both continuous and discrete transformations. The dynamics is captured by rewriting rules and an equational theory at the system level that describes how agents interact.
In our example, for instance, each move of a robot is synchronous with an effect on the field. 
Each agent therefore specifies how the action affects its state, and the system specifies which composite actions (i.e., set of simultaneous actions) may occur. 
We give hereafter an intuitive example that abstracts from the underlying algebra of each agent.

\paragraph{Agent} 
A robot and a field are two examples of an agent that specifies a component as a rewriting theory.
The dynamics of both agents is captured by a rewrite rule of the form:
\[
      (s,\emptyset)\Rightarrow (s',\acts)
\]
where $s$ and $s'$ are state terms, and $\acts$ is a set of actions that the field or the robot proposes as alternatives.
Given an action $a\in \acts$ from the set of possibilities, a function $\phi$ updates the state $s$ and returns a new state $\phi(s',a)$.
The equational theory that specifies $\phi$ may capture both discrete and continuous changes.
The robot and the field run concurrently in a system, where their actions may interact.

\begin{example}[Battery]
  A battery is characterized by a set of internal physical laws that describe the evolution of its energy profile over time under external stimulations. We consider three external stimuli for the battery as three events: a charge, a discharge, and a read event. Each of those events may change the profile of the battery, and we assume that in between two events, the battery energy follows some fixed internal laws.\\
  Formally, we model the energy profile of a battery as a function $f : \Rp \to [0,100\%]$ where $f(t) = 50\%$ means that the charge of the battery at time $t$ is of $50\%$. 
  In general, $f$ may be arbitrarily complex, and captures the response of event occurrences (e.g., charge, discharge, read) and passage of time coherently with the underlying laws (e.g., differential equation). For instance, a charge (or discharge) event at a time $t$ coincides with a change of slope in the function $f$ after time $t$ and before the next event occurrence.\\
  For simplicity, we consider a battery for which $f$ is piecewise linear in between any two events.
  The slope changes according to some internal laws at points where the battery is used for charge or discharge.\\
  In our model, a battery interacts with its environment only at discrete time points.
Therefore, we model the observables of a battery as a function $l : \N \to [0,100\%]$ that intuitively samples the state of the battery at some monotonically increasing and non-Zeno sequence of timestamp values.
We capture, in Definition~\ref{def:component}, the continuous profile of a battery as a component whose behavior contains all of such increasing and non-Zeno sampling sequences for all continuous functions $f$.
\end{example}
\begin{example}[Robot]
    A robot's state contains the previously read values of its sensors.
    Based on its state, a robot decides to move in some specific direction or read its sensors.\\
    Similarly to the battery, we assume that a robot acts periodically at some discrete points in time, such as the sequence $\move(E)$ (i.e., moving East) at time $0$, $\rread((x,y),l)$ (i.e., reading the position $(x,y)$ and the battery level $l$) at time $T$, $\move(W)$ (i.e., moving West) at time $3T$ while doing nothing at time $2T$, etc. The action may have as effect to change the robot's state: typically, the action $\rread((x,y),l)$ updates the state of the robot with the coordinate $(x,y)$ and the battery value $l$.
\end{example}

\paragraph{System}
A system is a set of agents together with a composability constraint $\kappa$ that restricts their updates. For instance, take a system that consists of a robot $\id$ and a field $F$. The concurrent execution of the two agents is given by the following system rewrite rule:
\[
    \{(s_\id,\acts_\id), (s_F, \acts_F)\} \Rightarrow_S  \{(\phi_\id(s_\id,a_\id),\emptyset), (\phi_F(s_F,a_F),\emptyset)\}
\]
where $a_\id \in \acts_\id$ and $a_F \in \acts_F$ are two actions related by $\kappa$.

Each agent is unaware of the other agent's decisions. The system rewrite $\Rightarrow_S$ filters actions that do not comply with the composability relation $\kappa$. As a result, each agent updates its state with the (possibly composite) action chosen at runtime, from the list of its submitted actions.
The framework therefore clearly separates the place where agent's and system's choices are handled, which is a source of runtime analysis.

Already, at this stage, we can ask the following query on the system: will robot $\id$ eventually reach the location $(x,y)$ on the field? Note that the agent alone cannot answer the query, as the answer depends on the characteristics of the field.
\begin{example}[Battery-Robot]
    \label{ex:battery-robot}
    Typically, a move of the robot \emph{synchronizes} with a change of state in the battery, and a read of the robot occurs at the same time as a sampling of the battery value.\\
    The system behavior therefore consists of sequences of simultaneous events occurring between the battery and the robot. 
    By composition, the battery exposes the subset of its behavior that conforms to the specific frequency of read and move actions of the robot. The openness of the battery therefore is reflected by its capacity to adapt to any observation frequency.
\end{example}

\paragraph{Coordination}
Consider now a system with three agents: two robots and a field. Each robot has its own objective (i.e., location to reach) and strategy (i.e., sequence of moves). Since both robots share the same physical field, some exclusion principals apply, e.g., no two robots can be at the same location on the field at the same time.
It is therefore possible that the system deadlocks if no actions are composable, or livelocks if the robots enter an infinite sequence of repeated moves.

We add a protocol agent to the system, which imposes some coordination constraints on the actions performed by robots $\id_1$ and $\id_2$.
Typically, a protocol coordinates robots by forcing them to do some specific actions.
As a result, given a system configuration $\{(s_{\id_1},\acts_{\id_1}), (s_{\id_2},\acts_{\id_2}), (s_F,\acts_F), (s_P, \acts_P)\}$ the run of robots $\id_1$ and $\id_2$ has to agree with the observations of the protocol, and the sequence of actions for each robot will therefore be conform to a permissible sequence under the protocol.

In the case where the two robots enter a livelock and eventually run out of energy, we show in Section~\ref{sec:implementation} the possibility of using a protocol to remove such behavior. 
\begin{example}[Safety property]
  A safety property is typically a set of traces for which nothing bad happens.
  In our framework, we consider only observable behaviors, and a safety property therefore declares that nothing bad \emph{is observable}. However, it is not sufficient for a system to satisfy a safety property to conclude that it is safe: an observation that would make a sequence violate the safety property may be absent, not because it did not actually happen, but merely because the system missed to detect it.\\
  For example, consider a product of a battery component and a robot with a sampling period $T$, as introduced in Example~\ref{ex:battery-robot}. 
  Consider the safety property: \emph{the battery energy is between the energy thresholds $e_1$ and $e_2$}. 
  The resulting system may exhibit observations with energy readings between the two thresholds only, and therefore satisfy the property. However, had the robot used a smaller sampling period $T'=T/2$, which adds a reading observation of its battery between every two observations, we may have been able to detect that the system is not safe because it produces sequences at this finer granularity sampling rate that  violate the safety property.
  We show how to algebraically capture the safety of a system constituted of a battery-robot. 
\end{example}
}

\iftoggle{short}{}{
\subsection{Instances of agents in Maude}
\label{appendix:agent}

    Section~\ref{sec:general-framework} introduces the signature for an agent module, and the rewrite rules for an agent and a system.
    The instance of an agent module provides an equational theory that implements each operation, namely \texttt{computeActions}, \texttt{internalUpdate}, \texttt{getOutput}, and \texttt{getPostState}.
    Each instance comes with an interface, called \texttt{AGENT-INTERFACE} in which the action names for \texttt{AGENT} are constructed.
    For instance, the interface for the robot agent called \texttt{TROLL}  contains the constructors for action names \texttt{move: direction -> Action} and \texttt{read: sensorName -> Action}.
    An agent that interacts with another agent must therefore include the interface module of that agent.
    We also assume that each agent shares the same preference structure, which we call action semiring (written \texttt{ASemiring}).
    The action semiring consists of an action paired with a natural number preference value.
     To illustrate the use of our framework to simulate and verify cyber-physical systems, we present an agent specification for four components: a \texttt{FIELD}, a \texttt{TROLL}, a \texttt{BATTERY}, and a \texttt{PROTOCOL}.

     A \texttt{FIELD} component interacts with the \texttt{TROLL} component by reacting to its move action, and its sensor reading.
     As shown in Listing~\ref{lst:field} the \texttt{FIELD} agent has no actions, but reacts to the move action of the \texttt{TROLL} agent by updating its state and changing the agent's location. 
     Currently, the update is discrete, but more sophisticated updates can be defined (e.g., changing the mode of a function recording the trajectory of the \sort{TROLL} agent).
     In the case where the state of the \sort{FIELD} agent forbids the \sort{TROLL} agent's move, the \sort{FIELD} agent enters in a disallowed state marked as \texttt{notAllowed(an)}, with \texttt{an} as the action name.
     The \sort{FIELD} responds to the read sensor action by returning the current location of the \sort{TROLL} agent as an output.
}

\begin{lstlisting}[caption={Extract from the \texttt{FIELD} Maude module.}\label{lst:field}]
fmod FIELD is
  inc TROLL-INTERFACE .
  inc FIELD-INTERFACE .
  inc PROTOCOL-INTERFACE .
  inc AGENT{ASemiring} .
  ...
  *** Passive agent:
  eq computeActions(id , Field, M) = null .
  
  eq internalUpdate(id , Field, M) = M .
  
  ceq getPostState(r, Field, id, a, mtOutput, M) = M'  
      if isMove?(a) /\
         k(loc) |-> d(id) , M1' := M /\
         loc' := next(loc, a) /\
         loc' =/= loc /\
         M[k(loc')] == undefined /\
         M' := k(loc') |-> d(id), M1' .

  ceq getPostState(r, Field, id, a, mtOutput, M) = notAllowed(a)
      if isMove?(a) /\ k(loc) |-> d(id) , M1' := M /\
         loc' := next( loc , a ) /\ ((loc' =/= loc and  M[k(loc')] =/= undefined) or loc' == loc) .

  ceq getOutput(r, Field, id, readSensors(position sn), M) 
                =  ( k("pos") |-> loc , M'  ) 
      if  k(loc) |-> d(id) , M1' := M /\ 
          M' := ( k("obstacles") |-> obstacle( 1, id, loc, M))   .
endfm
\end{lstlisting}

\iftoggle{short}{}{
     A \texttt{TROLL} agent reacts to no other agent actions, and therefore does not include any agent interface.
     However, the \texttt{TROLL} agent returns a ranked set of actions given its state with the \texttt{computeActions} operation. The expression may contain more than one action, with different weights. The weights of the action may depends on the internal goal that the agent set to itself, as for instance reaching a location on the field.
     The \texttt{TROLL} agent specifies how it reacts to, e.g., the sensor value input from the field, by updating the corresponding key in its state with \texttt{getSensorValues}.
}
\begin{lstlisting}[caption={Extract from the \texttt{TROLL} Maude module.}\label{lst:troll}]
fmod TROLL is
  inc AGENT{ASemiring} .
  inc LOCATION .
  inc TROLL-INTERFACE .

  eq computeActions(id , Troll , M ) = getSoftActions(id, M , trollActions(id, M)) .
  
  ceq internalUpdate(id, Troll, M) = insert(k("read"), nd(1), M) if M[k("read")] == nd(0) .
  ceq internalUpdate(id, Troll, M) = insert(k("read"), nd(0), M) if M[k("read")] == nd(1) .
  
  ceq getPostState(id, Troll, id, readSensors(sn), sensorvalues, M) = M'
      if M' := getSensorValues(getResources(id, readSensors(sn)) , sensorvalues), k("goal") |-> M[k("goal")], k("read") |-> nd(1)  .
endfm
\end{lstlisting}

\iftoggle{short}{}{
     A \texttt{BATTERY} agent does not act on any other agent, as the \texttt{FIELD}, but reacts to the \texttt{TROLL} agent actions.
     Each \texttt{move} action triggers in the \sort{BATTERY} agent a change of state that decreases its energy level. As well, each \texttt{charge} action changes the \sort{BATTERY} agent state to increase its energy level.
     Similarly to the field, in the case where the state of the battery agent has $0$ energy, the battery enters a disallowed state marked as \texttt{notAllowed(an)}, with \texttt{an} as the action name.
     A sensor reading by the \texttt{TROLL} agent triggers an output from the \sort{BATTERY} agent with the current energy level.
}
\begin{lstlisting}[caption={Extract from the \texttt{battery} Maude module.}\label{lst:battery}]
fmod BATTERY is
 inc AGENT{ASemiring} .
 inc BATTERY-INTERFACE . 
 inc TROLL-INTERFACE .


 *** Passive agent:
 eq computeActions(id, Battery, M ) = null .
 eq internalUpdate(id, Battery, M ) = M .


 ceq getOutput(r, Battery, id, readSensors(energy sn), M) 
           =  k("bat") |-> M[k("bat")] 
     if r := getBattery(id) .

 *** Next state.
 ceq getPostState(r, Battery, id, an, mtOutput, M) =  M'
     if isMove?(an) /\
        k("bat") |-> nd(s i) , M1' := M /\
        M' := insert( k("bat") , nd(i) , M) . 

 ceq getPostState(r, Battery, id, charge(j), mtOutput, M)  = M1 
     if nd(i) := M[k("bat")] /\ 
         i  < capacity  /\
         M1 := insert( k("bat") , nd(min ( i + j, capacity)) , M ) .

 ceq getPostState(r, Battery, id, an, mtOutput, M) = notAllowed(an)
      if isMove?(an) /\ M[k("bat")] == nd(0) .

endfm
\end{lstlisting}

A \texttt{PROTOCOL} agent \sort{swap(id1,id2)} acts on the \texttt{TROLL} agents \sort{id1} and \sort{id2}, and is used as a resource by the two \sort{TROLL} agent move action.
A \texttt{PROTOCOL} internally has a finite state machine \sort{T(id):Fsa} that accepts or rejects a sequence of actions. 
Each \texttt{move} action of a \sort{TROLL} is accepted only if there is a transition in the \sort{PROTOCOL} agent state transition system.
A \texttt{PROTOCOL} agent \sort{swap(id1, id2)} always tries to swap agents with ids \sort{id1} and \sort{id2}.
Thus, if \sort{id2} is on the direct East position of \sort{id1} on the field, then action \sort{start} succeeds, and the protocol enters in the sequence \sort{move(N)} for \sort{id2}, \sort{move(W)} for \sort{id2}, \sort{move(E)} for \sort{id1}, and then \sort{move(S)} for \sort{id2}. Eventually the sequence ends with \sort{finish} action.
The \sort{PROTOCOL} agent may also have some transitions labeled with a set of actions, one for each of the agent \sort{id1} and \sort{id2}. 
In which case, the transition succeeds if the clique contains, for each agent involved in the protocol, an action that is composable with the action labeling the protocol transition.
We use the \sort{end} action to mark the end of the sequence of actions forming a clique. The \sort{PROTOCOL} may reject such \sort{end} action if the clique does not cover the set of actions labeling the transition, which therefore discard the set of actions as not composable.
\begin{lstlisting}[caption={Extract from the \texttt{battery} Maude module.}\label{lst:swap}]
fmod SWAP is
    inc AGENT{ASemiring} .
    inc TROLL-INTERFACE .
    inc PROCESS-INTERFACE .
    inc FIELD-INTERFACE .
    inc PROTOCOL-INTERFACE .

    op T : Identifier -> Fsa .
    *** Update of state from external move or its own swapping actions
    ceq getPostState(id, Protocol, id', move(d), sysState, M ) = M'
        if  {q(i)} := getState(M) /\ 
            M' :=  insert( k("recv") , recv(union(getLabel(M), {l(id', move(d))})) , M) .

    *** Ending transition correctly
    ceq getPostState(id , Protocol , id, end , sysState, M) = M'
        if  state := getState(M) /\
            label := getLabel(M) /\
            tr := getTransitions(T(id)) /\
            (state, label, state'), tr' := tr /\
            M' :=  insert( k("recv") , recv({}), insert( k("state") , ds(state') , M)) .

    *** Not allowed states
    eq getPostState(id, Protocol, id', end, sysState, M ) = notAllowed(end) [owise] .
    eq getPostState(id, Protocol, id', a, sysState, M) = M [owise] .

    eq getOutput(id, Protocol, id', a, M) = empty .
    ceq computeActions(swap(id, id') , Protocol, M ) = ((swap(id, id') , ( start ; getResources(swap(id, id'), start))), 5) 
        if {q(0)} := getState(M) .
    eq computeActions(swap(id, id'), Protocol, M) = null [owise] .
    eq internalUpdate(swap(id, id'), Protocol, M) = M .

endfm
\end{lstlisting}

\iftoggle{short}{}{
\paragraph{Composability relation}
The \sort{TROLL}, \sort{FIELD}, and \sort{BATTERY} modules specify the state space and transition functions for, respectively, a \sort{TROLL}, \sort{FIELD}, and \sort{BATTERY} agent. A system consisting of a set of instances of such agents would need a composability relation to relate actions from each agent.\\
More precisely, we give some possible \emph{cliques} of a system consisting of two \sort{TROLL} agents with identifiers  \sort{id(0), id(1):TROLL}, one \sort{field:FIELD} agent, and two \sort{BATTERY} agents \sort{bat(0), bat(1):BATTERY}. \\
The actions of agent \sort{id(0)} compose with outputs of its corresponding battery \sort{bat(0)} and of the shared \sort{field} agent.\\
For instance, a move action of the \sort{id(0)} agent is of the form \sort{(id(0), (move(d), \{bat(0), field\}))}, where \sort{d} is a direction for the move, and composes with outputs of the battery and field, both notifying that the move is possible.\\
Alternatively, a read action of the \sort{id(0)} agent is of the form \sort{(id(0), (read, \{bat(0), field\}))} and composes with outputs of the battery and field, each giving the battery level and the location of agent \sort{id(0)}.

\paragraph{System} 
The agents defined above are instantiated within the same system to study their interactions.
We consider a system containing two \texttt{TROLL} agents, with identifiers \texttt{id(0)} and \texttt{id(1)}, paired with two \sort{BATTERY} agents with identifier \texttt{bat(0)} and \texttt{bat(1)}, and sharing the same \texttt{FIELD} resource. 
The goal for each agent is to reach the initial location of the other agent. If both agents follow the shortest path to their goal location, there is an instant for which the two agents need to swap their positions.
The crossing can lead to a livelock, where agents move symmetrically until the energy of the batteries runs out.%
}
The initial system term, without the protocol, is given by:
  \begin{lstlisting}
eq init = [[id(0): Troll | k("goal") |-> (5 ; 5) ; false ; null]
[bat(0) : Battery | k("bat") |-> nd(capacity) ; false ; null ]
[id(1): Troll | k("goal") |-> (0 ; 5) ; false ; null]
[bat(1) : Battery | k("bat") |-> nd(capacity) ; false ; null ]
[field : Field | (k(( 0 ; 5 )) |-> d(id(0)) , k(( 5 ; 5 )) |-> d(id(1))) ; false ; null]] .
  \end{lstlisting}

  The initial system term with the protocol is given by:

\begin{lstlisting}
eq init = [[id(0): Troll | k("goal") |-> (5 ; 5) ; false ; null]
[bat(0) : Battery | k("bat") |-> nd(capacity) ; false ; null ]
[id(1): Troll | k("goal") |-> (0 ; 5) ; false ; null]
[bat(1) : Battery | k("bat") |-> nd(capacity) ; false ; null ]
[swap(id(0),id(1)) : Protocol | k("state") |-> ds({q(0)}), k("recv") |-> recv({}) ; false ; null]
[field : Field | (k(( 0 ; 5 )) |-> d(id(0)) , k(( 5 ; 5 )) |-> d(id(1))) ; false ; null]] .
\end{lstlisting}
\end{document}